\documentclass[runningheads]{llncs}

\newcommand{\ignore}[1]{}
\newcommand{\angl}[1]{\langle #1 \rangle}
\newcommand{\sem}[1]{{[\![ #1 ]\!]}}
\newcommand{\RR}{\mathbb{R}}
\newcommand{\compclass}[1]{\textsf{#1}}
\newcommand{\NCone}{\compclass{NC\textsuperscript{1}}}
\usepackage{listings}
\usepackage{refcount}
\usepackage{algorithm,algorithmic}
\usepackage{amsfonts}
\usepackage{amsmath}
\usepackage[numbers,sort&compress]{natbib}
\usepackage[T1]{fontenc}
\usepackage{graphicx}
\usepackage{hyperref}

\begin{document}
\title{Max-Policy Iteration, Revisited}
%
%
\author{David Monniaux\inst{1}\orcidID{0000-0001-7671-6126} \and
Helmut Seidl\inst{2}\orcidID{0000-0002-2135-1593}}
\authorrunning{D.\ Monniaux and H.\ Seidl}
%
\institute{CNRS/VERIMAG, Universit{\'e} Grenoble Alpes,
38401 Saint Martin d’Hères, France \\
\email{david.monniaux@univ-grenoble-alpes.fr} \\
\url{https://www-verimag.imag.fr/~monniaux/}
\and
School of CIT, TU Munich, 85748 Garching, Germany\\
\email{helmut.seidl@tum.de}\\
\url{https://www.cs.cit.tum.de/en/pl/personen/helmut-seidl/}}
\maketitle              
\begin{abstract}
\emph{Max}-policy iteration is an approach to computing precise numeric program invariants by successive attempts at resolving maximum operators and reduction to mathematical optimization.
Mathematical optimization, though, may be expensive.
Here, we show, for \emph{max}-policy iteration on systems of equations over integers as well as over floating point numbers, that mathematical optimization can be replaced by plain value iteration --- which is still guaranteed to terminate. 
As an application, a precise bound analysis for integer or floating point variables is obtained, avoiding widening operators altogether.

We also consider \emph{max}-policy iteration over the rational numbers, where the right-hand sides are maxima of minima of affine combinations of unknowns. 
We propose \emph{min}-policy iteration as an alternative to linear programming for solving the optimization problems posed by \emph{max}-policy iteration.
We prove that \emph{max-min} policy iteration is guaranteed to return the least solution for bounded systems.
We also show how to extend this algorithm to unbounded systems, and how to construct certificates of soundness as well as of optimality of the computed results.

\keywords{policy iteration \and interval analysis \and certificates}
\end{abstract}
\section{Introduction}\label{sec:intro}
%
Program analysis by abstract interpretation aims at inferring invariants (numeric or otherwise) 
inside restricted classes of invariants known as \emph{abstract domains}.
The most common approach is to compute an ascending sequence of invariant candidates until one is inductive (stable by ``next state transition''), an approach known as \emph{value iteration} in other fields (e.g., in the analysis of probabilistic games).
Such ascending sequences may fail to converge within finite time, or may converge, but not within tolerable finite time.
A generic remedy for these cases is to apply \emph{widening operators} to accelerate convergence \cite{Cousot77,DBLP:journals/logcom/CousotC92}. 
Widening, however, suppresses two interesting properties of value iteration:
\begin{itemize}
\item 	value iteration, when it converges and \emph{precise} abstract transfer functions are used, 
	provides the most precise invariant in the chosen abstract domain;
\item 	the result of value iteration, assuming monotone transfer functions, 
	is monotone in the preconditions and other assumptions on the program used by the analysis: 
	more precise information thus results in more precise invariants.
\end{itemize}
Lack of monotonicity is a form of \emph{brittleness}: the analyzer may have different behaviors for  
closely related programs even though users would expect it to be able to infer the same invariants; 
it may, counter-intuitively, produce worse invariants if a more precise precondition is given.

In response to these weaknesses, two approaches inspired by game theory were proposed for inferring 
numerical invariants, including intervals of variation: \emph{max}-policy iteration (or \emph{ascending} policy iteration) and \emph{min}-policy iteration (or \emph{descending} policy iteration).%
\footnote{Following game theory terminology, a \emph{policy} may also be known as a \emph{strategy}. 
In game theory, \emph{max} strategies express possible strategies for a player seeking to maximize 
the result of the game, fighting against a player that seeks to minimize it. Policy iteration may be used 
to find the values of games~\cite{Puterman_1994}.}
These techniques are applicable when the invariants to be considered can be expressed by means of 
\emph{linear} orders.
The systems to be solved, consist of equations of the form $x = e$, $x$ an unknown, 
$e$ a right-hand side expression built up from other unknonws of the system and constants by means of operators.
Besides arithmetic, also minimum and maximum operators may be employed
where maximum operators correspond to control-flow joins, and minimum operators to guards and 
relations between constraints.

%
The application of a \emph{max}-policy consists in replacing each occurrence of a maximum operator 
(here denoted by $\vee$) by one of its argument expressions to obtain
a system without maximum operators.
\emph{Max}-policy iteration over the reals
may rely on mathematical optimization (e.g., linear or convex programming solvers)
to solve the simplified systems
\cite{Gawlitza_Seidl_FMSD2013,GawlitzaSAGG12}.
That solution is then used either to \emph{improve} the current \emph{max}-policy or 
identify that a solution to the original system has been found.
Furthermore, when each occurring simplified system has a \emph{unique} solution,
the result is guaranteed to be the \emph{least} solution of the system.

\ignore{
%
%
\emph{Max}-policy iteration, since it is a form of acceleration of ascending iterations, may be combined 
with other forms of abstract interpretation. If it is implemented exactly (no relaxation, no floating-point 
approximation $\ldots$), it computes abstract states that would be reachable by value iteration, thus it can be combined without any loss of precision with value iteration on other domains. Of course, we lose optimality properties 
if it is combined with iteration with widening.
}

The application of a \emph{min}-policy consists in replacing each occurrence of a minimum operator (here denoted as $\wedge$)
with one of its argument expressions
\cite{DBLP:conf/cav/CostanGGMP05,DBLP:conf/esop/GaubertGTZ07,AdjeGG2012}.
Again, the solution of the simplified system is used either to detect termination or
improve the current \emph{min}-policy.
This iteration on \emph{min}-policies will result in a sequence of \emph{over-approximations} 
of the system solution yielding a \emph{descending} sequence of inductive invariants. 
There are classes of well-behaved problems 
where it can be proven to return the least solution~\cite{adje:hal-00940804}.
In general, though, this need not be the case.

\paragraph{Contributions.}
In this paper, we re-consider \emph{max}-policy iteration. 
We identify two practical cases, namely, when the value domain either 
consists of integers or of floating-point numbers, where \emph{max}-policy iteration combined with value 
iteration is guaranteed to terminate (Section \ref{sec:max}). 
This termination result provides us with a technique for inferring numeric program invariants 
which neither resorts to widening nor to mathematical optimization.
Over the integers, the new method is also applicable when current methods fail, e.g., because 
right-hand sides use operations such as integer division.
In case of floating-point numbers, the new method vastly outperforms 
methods based on linear programming. 

Value iteration for the simplified systems encountered during \emph{max}-policy iteration, though,  is not sufficient to solve systems using \emph{exact} arithmetic over the rationals.
%
Assume all right-hand side expressions of equations are
of the form $\bigvee_{i\in I}\bigwedge_{j\in J_i} e_{ij}$ where each $e_{ij}$ is an
affine combination of unknowns.
We propose to combine \emph{max}- and \emph{min}-policy iteration --- 
to reduce solving of the original system to solving a sequence of affine systems of equations
(Section \ref{sec:max-min}.
To deal with unbounded solutions, that basic setup may further be
enhanced with \emph{symbolic} upper bounds (Section \ref{sec:infinity}).
%
%
This algorithm differs from existing approaches in two ways:
\begin{itemize}
\item 	Contrary to \emph{min}-policy iteration, it is guaranteed to provide the \emph{least}  solution of the system.
\item 	Contrary to the known approach for \emph{max}-policy iteration over the rationals, it does not have to solve a succession of large linear programs. 
%
\end{itemize}
We also explain (Section~\ref{sec:witness}) how to provide \emph{easily checkable} witnesses 
for correctness and optimality of the result, and study them from a complexity-theoretical point of view.
%

In the case of programs to be analyzed that compute over floating-point values, 
it is desirable to implement the analyzer itself using floating-point (instead of attempting to 
over-approximate the behavior of the floating-point program using 
exact arithmetic). 
In Section~\ref{sec:floating-point}, we experimentally compare three variants \textbf{Max+X} of 
\emph{max}-policy iteration for
floating-point numbers that differ in the way how the encountered systems without maximum operators are solved:
by linear programming (\textbf X = \textbf{LP}), 
by a variant of \emph{min}-policy iteration (\textbf X = \textbf{Min}), 
or by (descending) value iteration (\textbf X = \textbf{Val}).
We discuss soundness and efficiency of the respective approaches (Section~\ref{sec:experiments}).
Section~\ref{sec:extensions} reports on extensions of our techniques to octagon and 
linear template constraint analysis.
%
%
\ignore{
Our implementation is available at \href{https://dx.doi.org/10.5281/zenodo.18134107}{\textsf{10.5281/zenodo.18134107}}.
}
%



\section{Systems of numerical equations} 
				\label{sec:monotone-systems}
				\label{sec:background}

In this section, we recall approaches for automatically inferring inductive 
invariants of programs and introduce an abstract framework for \emph{max}-policy iteration.
As a running example, we will use the following program for which we want 
to infer an inductive invariant of the form $\lstinline|i| \leq b$ at the loop head:
%
{\small\begin{lstlisting}
          int i=0;
/* 1: */  while (i < N) {
/* 2: */     i = i+1;
/* 3: */  }
/* 4: */
\end{lstlisting}}
where $N$ is some given integer constant. 
\emph{Interval analysis} \cite{Cousot77} determines for each program variable
$i$ at each program point $u$ an interval containing all values possibly attained at $u$ for $i$. 
Following, e.g., \cite{DBLP:conf/cav/CostanGGMP05,GawlitzaS2007}, 
we consider upper bounds and lower bounds of intervals separately.
We introduce for each program point $v$ and each program variable $i$ the unknowns $\angl{v,i^+}$ and $\angl{v,i^-}$ 
to represent an upper bound and \emph{negated} lower bound, respectively, to the values of $i$ 
whenever program point $v$ is reached. 
%
Then right-hand side expressions are provided which formalize how the value of $i$ at program point $v$
may depend on the values of variables at control-flow predecessors of $v$.
The corresponding equation system for our example program is summarized in Figure \ref{fig:intervals}.
Here, 1, 2, 3 and 4 represent the program points 
before the loop head, 
at the beginning and end of the loop body, and at program exit, respectively.
The initialization of \lstinline|i| at program start assigns the value 0 to $i$; likewise, all values of $i$
at the end of the loop body may flow into the variable $i$ at program point 1, 
resulting in the equations for $\angl{1,i^+}$ and $\angl{1,i^-}$ in the first row 
(the operator ``$\vee$'' denoting maximum).
%
The guard $i < N$ in the loop head constrains the values of $i$ by with program point
$2$ of the loop body is reached.
It is interpreted as the minimum (operator ``$\wedge$'') of the upper bound of $i$ before the guard and $N-1$.
%
To ensure that the upper bound for $i$ at program point 2
subsumes the current lower bound for $i$, we introduce the operator $?$.
The operator application $a\,?\,b$ returns $b$, if $a\geq 0$ and $-\infty$ otherwise.
Thus in row 2, the negated lower bound for $i$ at program point 2 is set to $-\infty$,
if the difference between $N-1$ and the lower bound for $i$ at program point 1 is negative.
Row 3 of the table corresponds to the increment of $i$, while
row 4 formalizes the condition of the negated guard onto the value of $i$ at program point 2 for reaching
the endpoint 4 of the program.

\begin{figure}[hbt]
\[
\begin{array}[t]{lll}
\angl{1,i^+}	&=&	0\vee\angl{3,i^+}	\\
\angl{2,i^+}	&=&	(\angl{1,i^-} +(N-1))? (\angl{1,i^+}\wedge (N-1)) \\
\angl{3,i^+}	&=&	\angl{2,i^+}+1	\\
\angl{4,i^+}	&=&	(\angl{1,i^+}-N)? \angl{1,i^+}
\end{array}
\quad
\begin{array}[t]{lll}
\angl{1,i^-}	&=&	0\vee\angl{3,i^-}	\\
\angl{2,i^-}	&=&	(\angl{1,i^-} +(N-1))? \angl{1,i^-} \\
\angl{3,i^-}	&=&	\angl{2,i^-}-1	\\
\angl{4,i^-}	&=&	(\angl{1,i^+}-N)? (\angl{1,i^-}\wedge - N)
\end{array}
\]
\caption{\label{fig:intervals}The bounds analysis for the example program.}
\end{figure}
\ignore{
The binary operator ``$?$'' on the other hand, is required to check that intersections of intervals 
are non-empty. It is defined by 
$n_1\,?\,n_2 = {-\infty}$ if $n_1 < 0$ and $n_2$ otherwise.
As for interval equations, we also require a sequence operator $;$ given by
$n_1;\,n_2 = {-\infty}$ if $n_1={-\infty} $ and $n_2$ otherwise.
}
Assuming that $N\geq 0$, the \emph{least} solution of this system is given by 
the table
\[
\begin{array}[t]{|l|r|}
\hline
\angl{1,i^-}	&	0	\\
\angl{1,i^+}	&	N	\\
\hline
\end{array}
\quad
\begin{array}[t]{|l|r|}
\hline
\angl{2,i^-}	&	0	\\
\angl{2,i^+}	&	N-1	\\
\hline
\end{array}
\quad
\begin{array}[t]{|l|r|}
\hline
\angl{3,i^-}	&	-1	\\
\angl{3,i^+}	&	N	\\
\hline
\end{array}
\quad
\begin{array}[t]{|l|r|}
\hline
\angl{4,i^-}	&	-N	\\
\angl{4,i^+}	&	N	\\
\hline
\end{array}
\]
\ignore{
-- from which the least solution of the interval system can be recovered as
\[
\begin{array}[t]{|l|l|}
\hline
\angl{1,i}	&	[0,N]	\\
\angl{2,i}	&	[0,N-1]	\\
\angl{3,i}	&	[1,N]	\\
\angl{4,i}	&	[N,N]	\\
\hline
\end{array}
\]
}
More generally, we consider a linearly ordered value domain $\mathbb D$, which could be, e.g.,
$\mathbb Z$ (the integers)
$\mathbb F$ (the floating point numbers), $\mathbb Q$ (the rationals)
or $\mathbb R$ (the reals). 
For $\mathbb Z$, $\mathbb Q$ and $\mathbb R$ we assume them to be
extended by $-\infty$ and $\infty$ as least and greatest elements, respectively.
All these examples (with the exception of $\mathbb Q$) are
not only linearly ordered, but \emph{complete} lattices.
\footnote{\label{note:qelim}%
The rationals do not form a complete lattice: $\{ x \mid x^2 \leq 2 \}$, e.g., has no rational upper bound. 
This is not an issue here as all arithmetic operations we consider are linear. The least solution
of a system of equations built from these together with maximum and minimum
can be expressed in the first-order theory of linear rational arithmetic. 
By applying quantifier elimination, a quantifier-free formula 
can be obtained defining exactly the desired solution~\cite{DBLP:conf/popl/Monniaux09}, thus proving
that it is indeed rational. 
A wide class of numeric static analysis problems can be solved by quantifier elimination 
\cite{DBLP:conf/sas/Monniaux07}: an approach, however, that does not scale.
%
}
Generally, we denote the ordering relation on the linear order $\mathbb D$ by ``$\leq$''.
Over $\mathbb D$, we consider systems $\Sigma$ of equations of the form
\begin{equation}
x = e_x\qquad(x\in X)
\label{eq:system}
\end{equation}
where $X$ is a finite set of \emph{unknowns}, and for each $x\in X$, the right-hand side $e_x$ is an expression
built up from constants from $\mathbb D$ and unknowns $y\in X$ by means of operators from some set $\Omega$. 
%
Each operator $\omega\in\Omega$ is a \emph{Scott-continuous} 
function ${\mathbb D}^k\to{\mathbb D}$ of some arity $k\geq 1$, i.e.,
commutes with least upper bounds of \emph{directed} subsets\footnote{
A subset $D$ is \emph{directed}, if for each $a,b\in D$, there is some $c\in D$ with
$a\leq c$ and $b\leq c$.}.
The assumption of Scott-continuity is also justified for floating-point numbers
$\mathbb F$ since that domain is finite, and common operators such as addition, 
multiplication with non-negative constants 
as well as rounding all are monotonic.
As usual, we write binary operators (such as ``$+$'') in infix notation.
One particular operator in $\Omega$ is $\vee$ (least upper bound). 
All operators in $\Omega$ different from $\vee$ are assumed to be \emph{strict},
i.e., preserve the least element $-\infty$. 
Since all operators are 
Scott-continuous, so is the \emph{evaluation} function
$\sem{e}:(X\to{\mathbb D})\to{\mathbb D}$ which, for every expression $e$ is defined by
\[
\begin{array}{lll}
\sem{c}\,\rho	&=& c	\\
\sem{x}\,\rho	&=& \rho(x)	\\
\sem{\omega(e_1,\ldots,e_k)}\,\rho	&=& \omega(\sem{e_1}\,\rho,\ldots,\sem{e_k}\,\rho)
\end{array}
\]
A mapping $\rho: X\to{\mathbb D}$ is a \emph{solution} 
to the equation system \eqref{eq:system} if 
$\rho(x) = \sem{e_x}\,\rho$
%
%
holds for all unknowns $x\in X$. 
\ignore{
Assume that $\mathbb D$ is a complete lattice.
Then two facts are worth noting.
\begin{enumerate}
\item	If $\rho:X\to\mathbb D$ is a pre-solution, then there is a least solution $\rho_1$ 
	with $\rho\leq\rho_1$; likewise,
\item	If $\rho:X\to\mathbb D$ is a post-solution, then there is a greatest solution $\rho_1$ 
	with $\rho_1\leq\rho$.
\end{enumerate}
The assignment $\rho_0$ which maps each unknown to the least possible value $-\infty$ in 
$\mathbb D$ clearly is a pre-solution. 
}
\emph{Value iteration} for a system starts from some valuation $\rho_0:X\to\mathbb D$, and successively computes 
a sequence $\rho_t, t\geq 0,$ where for $t>0$, $\rho_t$ is obtained from $\rho_{t-1}$ by updating the value
of $\rho_{t-1}$ for some unknown $x$ with $\sem{e_x}\,\rho_{t-1}$. 
In case that $\rho_0\leq\underline\rho$, $\underline\rho$ the least solution of the system, 
the sequence is \emph{increasing}.
Dually, when $\rho_0\geq\overline\rho$ where $\overline\rho$ is the greatest solution of the system, 
the sequence is \emph{decreasing}.
This sequence, be it increasing or decreasing, may be very long or, 
in case of the integers, the rationals or reals even infinite, 
i.e., never reach a solution.
%
%

Consider our example program. Simplifying the system of equations for lower and upper 
bounds for \lstinline|i|, and restricting to it to the upper bound for \lstinline|i|
at the loop head, while optimistically assuming the value 0 as a lower bound to \lstinline|i|, 
provides us with the single equation 
\begin{equation}
  x = 0 \vee ((x \wedge (N-1))+1)\label{eqn:simple_loop}
\end{equation}
where we have written $x$ instead of $\angl{1,i^+}$.
Assuming $N > 0$, value iteration iteratively computes for $x$ 
the values $0$, $1$, \dots, until reaching~$N$.
To speed up iteration, a practical analyzer therefore may apply \emph{widening} 
\citep{DBLP:journals/logcom/CousotC92} by extrapolating
the value for $x$ to $+\infty$.
%
A subsequent \emph{narrowing} pass, starting with the result from the previous iteration,
may in this case recover the value $N$ for $x$.
Consider now the following slightly more involved program:
{\small\begin{lstlisting}
int i=0;
while (true) {
  if (*) {
    if (i < N) i = i+1;
    else i = i/2;
  }
}
\end{lstlisting}}
\noindent where \lstinline|(*)| expresses nondeterministic choice. 
A corresponding (slightly) simplified equation for the
upper bound of \lstinline|i| at the loop head now is
\begin{equation}
  x = 0 \vee x\vee ((x \wedge (N-1))+1)\vee x/2 \label{eqn:simple_loop2}
\end{equation}
The valuation attained after widening again assigns the value $\infty$ to $x$.
Starting from there, a narrowing iteration retains this value --
although the least solution of the equation is still $N$.
Various improvements on widening \& narrowing would work on this example 
and provide the value $N$ for $x$.
Yet, they are still heuristic, there is no guarantee that they would work on other examples and produce the least invariant in the domain \cite{DBLP:conf/sas/GopanR07}. 
Our goal here is to describe a \emph{principled} approach that is applicable to a large class of cases.

\section{\emph{Max}-policy iteration}
	\label{sec:max-policies}\label{sec:max}

\emph{Max}-policy iteration \citep{GawlitzaS2007,GawlitzaS2011,Gawlitza_Seidl_FMSD2013} 
works by repeatedly replacing each $\vee$-operator in right-hand sides with one of its arguments
according to the current policy.
The iteration results in
succession of systems without $\vee$, whose least solutions \emph{under-approximate} the least solution 
of the original problem.
\ignore{
For simplicity, we present \emph{max}-policy iteration as applied to one system of equations
which could be thought of as constructed for the interval analysis of a program.
%
%
The construction of this system may have introduced operators ``$;$'' (sequence) and ``$?$''
which can conveniently be resolved during policy iteration.
For a clearer exposition, we therefore omit the treatment of these operators
from our technical discussion.
We also mention in passing that policy iteration need not be performed monolithically for 
a system of equations representing the whole program --- but locally as part of a generic 
abstract interpretation solver where it may
collaborate with other domains \citep{Karpenkov_PhD,karpenkov:hal-01255314}. 
Policy iteration then appears as a kind of \emph{magical widening operator}.
}
Consider a system of equations \eqref{eq:system} where right-hand sides $e_x$ are of the form
\begin{equation}
e_x = \mbox{$\bigvee_{j\in J_x} e_{x,j}$}\qquad(x\in X)
\label{rhs:max}
\end{equation}
Here, $J_x\subseteq{\mathbb N}_0$ is a finite set of \emph{choices} where each sub-expression $e_{x,j}$
does not have any further occurrences of $\vee$. 
%
For each choice $j$ of $x$, $\sem{e_{x,j}}:(X\to{\mathbb D})\to{\mathbb D}$ still is Scott-continuous.
For convenience, we assume that $0\in J_x$ for each $x\in X$ where $e_{x,0} = {-\infty}$, 
i.e., evaluates to the least element in the linear order.

A \emph{max}-policy $\pi:X\to{\mathbb N}_0$ assigns to each variable $x$ a possible choice $\pi(x)\in J_x$.
For $\pi$, we obtain from $\Sigma$ the simplified system $\Sigma^\pi$ as
\begin{equation}
x = e_{x,\pi(x)}\qquad(x\in X)
\label{eq:simpified_system}
\end{equation}
Given that 
$\mathbb D$ is complete, this system has a \emph{greatest} solution --- which is also the case for the
rationals when each $e_{x,\pi(x)}$ is a minimum of \emph{affine} expressions.\textsuperscript{\ref{note:qelim}}
Here, an affine expression is built up from constants and unknowns by addition and multiplication with 
non-negative constants. 
%
In a nutshell, \emph{max}-policy iteration proceeds as follows (see Fig.\ref{fig:max}). 
\begin{itemize}
\item
It starts with the policy $\pi_0 = \{x\mapsto 0\mid x\in X\}$ and $\rho_0 = \{x\mapsto-\infty\mid x\in X\}$.
\item
For each $t\geq 0$, it determines a mapping $\rho'_{t+1}$ with $\rho'_{t+1}(x) = \sem{e_{x}}\,\rho_t$.
It checks whether $\rho_t$ is already a solution of $\Sigma$,
i.e., whether $\rho_t = \rho'_{t+1}$.
%
%
\item
If not, a new policy $\pi_{t+1}$ is determined \emph{reluctantly} from $\pi_t$.
This means, that $\pi_{t+1}(x) = \pi_t(x)$ whenever $\rho_t(x) = \rho'_{t+1}$ holds.
Otherwise, a $j$ is selected for $x$ where the maximum of the values
$\sem{e_{x,j}}\,\rho_t$ is attained.
\item
For each encountered policy $\pi_{t+1},t\geq 0$, it determines a solution $\rho_{t+1}\geq\rho'_{t+1}$
of the simplified system $\Sigma^{\pi_{t+1}}$ by means of the external solver \textbf{solve}. 
\end{itemize}
\begin{figure}[t]
\begin{center}
\[
\begin{array}{l}
\textsf{max\_policy\_iterator}\;(\textbf{solve},\Sigma)\;\{	\\
\qquad\begin{array}[t]{l}
\pi\;{:=}\;\{x\mapsto 0\mid x\in X\};\\	
\rho\;{:=}\;\{x\mapsto -\infty\mid x\in X\};\\
\textbf{do}\;\{	\\
\qquad\rho'\;{:=}\;\{x\mapsto \sem{e_x}\,\rho\mid x\in X\};	\\
\qquad\textit{modified}\;{:=}\;\textbf{false};	\\
\qquad\textbf{forall}\;(x\in X)	\\
\qquad\qquad\textbf{if}\;(\rho'(x) \neq \rho(x))\;\{ \\
\qquad\qquad\qquad\textit{modified}\;{:=}\;\textbf{true};	\\
\qquad\qquad\qquad\pi(x)\;{:=}\;\textbf{argmax}_{J_x}\{j\mapsto \sem{e_{x,j}}\,\rho\};	\\
\qquad\qquad\}	\\
\qquad\rho = \textbf{solve}(\Sigma^\pi,\rho');	\\
\}\;\textbf{while}\,(\textit{modified});\\
\textbf{return}\;\rho;
\end{array} 
\end{array} 
\]
\caption{\label{fig:max}\emph{max}-policy iteration algorithm using \textbf{solve} to
determine solutions of $\Sigma^\pi$.}
\end{center}
\end{figure}
Consider equation \eqref{eqn:simple_loop2} for the second example program, 
decorated with $-\infty$:
\[
x = {-\infty}\vee 0\vee x \vee (x\wedge (N-1)) +1 \vee x/2
\]
Assume that \textbf{solve} always determines the \emph{greatest} solutions.
Then the following policies, simplified systems and solutions are encountered:
\[
\begin{array}{l|c|l|r}
t\;\;	& \;\;\pi_t(x)\;\;	&\;\; \Sigma^{\pi_t} & \;\;\rho_{t}(x)	\\
\hline
\hline
0	& 0 & \;\;x = -\infty & -\infty	\\
1	& 1 & \;\;x = 0	& 0		\\
2	& 3 & \;\;x = (x\wedge (N-1)) +1\;\; & N\\
\end{array}
\]
Note that, due to reluctant improvement, neither policy $\{x\mapsto 2\}$ nor $\{x\mapsto 4\}$ ever is selected.
In general, we have:

\begin{lemma}\label{l:term}

\emph{Max}-policy iteration for systems of equations with right-hand sides 
\eqref{rhs:max}, using a function \textbf{solve} to compute \emph{greatest} solutions, 
terminates with a solution.
\end{lemma}

\begin{proof}
As the set $X$ of unknowns as well as the sets $J_x$, $x\in X$, are finite,
it suffices to prove that the same \emph{max}-policy is never encountered twice.
For a contradiction, assume that the sequence $\pi_t,t\geq 0,$ of encountered policies
is infinite.
Consider a \emph{max}-policy $\pi_t$ encountered during the iteration with 
$\rho_t$ as the greatest solution of the simplified system $\Sigma^{\pi_t}$.
Assume that $\pi_t$ is improved to a 
\emph{max}-policy $\pi_{t+1}$. This means that there is a non-empty subset $Y\subseteq X$
such that $\sem{e_{x,\pi_{t+1}}}\,\rho_t = \sem{e_{x}}\,\rho_t > \rho_t(x)$ for all $x\in Y$ and
$\pi_{t+1}(x) = \pi_t(x)$ otherwise.
%
Due to monotonicity of all $\sem{e_{x,j}}$, $\Sigma^{\pi_{t+1}}$ has a \emph{least}
solution $\underline\rho_{t+1} \geq \rho_t$ with $\underline\rho_{t+1}(x) > \rho_t(x)$ for all $x\in Y$.
The solution $\underline\rho_{t+1}$ of $\Sigma^{\pi_{t+1}}$, however, is bounded by the \emph{greatest} 
solution $\rho_{t+1}$ of $\Sigma^{\pi_{t+1}}$. Accordingly,
$\rho_t\leq\rho_{t+1}$ with $\rho_t(x) < \rho_{t+1}(x)$ for every $x\in Y$.
The sequence $\rho_t,t\geq 0$, therefore is strictly increasing --- implying that all systems
$\Sigma^{\pi_t}$ must be distinct. Since the number of \emph{max}-policies is finite, this is impossible.
\qed
\end{proof}

\noindent
Let us call each system $\Sigma^{\pi_{t+1}}$ encountered during \emph{max}-policy iteration for $\Sigma$
a \emph{\textit{max}-iterate}. A max-iterate $\Sigma^{\pi_{t+1}}$ is deemed \emph{concave}
if the greatest solution $\rho_{t+1}$ of $\Sigma^{\pi_{t+1}}$ 
is the \emph{only} solution $\underline\rho$ of $\Sigma^{\pi_{t+1}}$ with
\begin{equation}
\sem{e_{x,\pi_{t+1}(x)}}\,\rho_t\leq\underline\rho(x)\qquad(x\in X)
\label{def:concave_iterates}
\end{equation}
Indeed, when the solutions $\underline\rho_{t+1}$ and $\rho_{t+1}$ in the proof of Lemma \ref{l:term} \emph{coincide}, 
\emph{max}-policy iteration will terminate with the \emph{least} solution of $\Sigma$. We have:

\begin{lemma}\label{l:unique}
If the system $\Sigma$ with right-hand sides 
\eqref{rhs:max} 
has only concave \emph{max}-iterates,
\emph{max}-policy iteration for $\Sigma$
terminates with the \emph{least} solution of $\Sigma$.
\end{lemma}

\begin{proof}
Let $\rho_t,t=0,\ldots,N$ denote the sequence of assignments computed by \emph{max}-policy
iteration, and $\rho^*$ the least solution of \eqref{eq:system}. 
By induction on $t\geq 0$, we prove that $\rho_t\leq \rho^*$ holds.
%
For $t=0$, this is trivially true. So assume that $\rho_t\leq\rho^*$, and
$\rho'_{t+1}(x) = \sem{e_x}\,\rho_t$.
Let $\pi_{t+1}$ the next \emph{max}-policy determined according to $\rho_t$ and $\rho'_{t+1}$.
By definition, $\rho_{t+1}$ equals the only solution $\underline\rho$ of $\Sigma^{\pi_{t+1}}$ 
with 
$
\rho'_{t+1}\leq\underline\rho
$.
%
Since expression evaluation is Scott-continuous, 
$\rho_{t+1}$ can be obtained as the least upper bound $\bigvee_{j\geq 0}\rho^{(j)}$
where $\rho^{(0)}(x) = \sem{e_{x,\pi_{t+1}}}\,\rho_t$ and for $j>0$, 
$\rho^{(j)}(x) = \sem{e_{x,\pi_{t+1}(x)}}\,\rho^{j-1)}$. 
We prove that for all $j\geq 0$, $\rho^{(j)}\leq\rho^*$.
For the base case $j=0$, we have
by monotonicity and the inductive hypothesis for $\rho_t$, 
\[
\rho^{(0)}(x) =
\sem{e_{x,\pi_{t+1}(x)}}\,\rho_t =
\sem{e_{x}}\,\rho_t \leq
\sem{e_{x}}\,\rho^* =\rho^*(x)
\]
holds for all $x\in X$.
For $j>0$, we assume by induction hypothesis, that $\rho^{(j-1)}\leq\rho^*$. 
Then
\[
\rho^{(j)}(x) = 
\sem{e_{x,\pi_{t+1}}}\,\rho^{(j-1)} \leq 
\sem{e_{x,\pi_{t+1}}}\,\rho^* \leq 
\sem{e_x}\,\rho^* = \rho^* (x)
\]
and the claim follows.
In particular, $\rho_N\leq\rho^*$.
Since $\rho_N$ is a solution of \eqref{eq:system}, we conclude that $\rho_N=\rho^*$ holds.
\qed
\end{proof}

\noindent
The assumption of Lemma \ref{l:unique} seems rather specific.
Still, only concave \emph{max}-iterates occur  in two important cases:
\begin{itemize}
\item 	for integer equations with addition, multiplication with non-negative constants
	and dedicated operators for multiplication of intervals
	\cite{DBLP:conf/cav/GawlitzaS09};
	%
	%
\item	for rational equations
	when right-hand sides use as arithmetic operators 
	addition as well as multiplication with 
	non-negative constants  --- where
	each right-hand side is \emph{bounded} by some finite bound $M<\infty$
	\cite{GawlitzaS2011}.
\end{itemize}
Let us call these two instances systems of \emph{polynomial integer equations} and 
\emph{bounded affine rational equations}, respectively.
%
As a first corollary, we thus recover the results from 
\cite{DBLP:conf/cav/GawlitzaS09,GawlitzaS2011}:

\begin{corollary}\label{c:least}
\emph{Max}-policy iteration terminates and returns least solutions 
for polynomial integer systems as well as for bounded affine rational systems. 
\qed
\end{corollary}

\noindent
Beyond these two cases, we can use \emph{max}-policy iteration also to determine \emph{some}
solution of a system of equations. 
Generally, when applying \emph{max}-policy iteration,
we are left with the task of computing the \emph{greatest} solutions of the encountered 
\emph{max}-iterates. An important result in this direction is:

\begin{theorem}\label{t:term}
For the complete linear orders $\mathbb Z$ and $\mathbb F$, 
the greatest solutions of \emph{max}-iterates
can be found by value iteration.
Therefore, \emph{max}-policy iteration combined with value iteration is an effective 
algorithm for solving systems with right-hand sides of the form \eqref{rhs:max}
over ${\mathbb D}\in\{{\mathbb Z},{\mathbb F}\}$.
\end{theorem}

\begin{proof}
Assume for a contradiction, that value iteration
does not terminate with computing the greatest solution $\rho_{t+1}$ for the $(t+1)$th encountered
\emph{max}-policy $\pi_{t+1}$. 
Then there is a variable $x$ for which a strictly decreasing sequence
of values $d_0 > d_1 > \ldots > d_k >\ldots$ is computed.
In case of integers, this implies that the greatest lower bound for $x$ is ${-\infty}$.
This means that not only $\rho_{t+1}(x) = {-\infty}$, but (due to monotonicity) also
$\rho_{t'}(x) = {-\infty}$ for the greatest solutions of the systems $\Sigma^{\pi_{i'}}$ for all
previously encountered max policies $\pi_{t'}$. Due to reluctant updates of policies, this
means that $\pi_{t+1}(x) = \pi_{0}(x) = 0$, i.e., the right-hand side of $x$ in $\Sigma^{\pi_{t+1}}$
already equals ${-\infty}$. Therefore, all $d_k$ equal ${-\infty}$ --- contradiction.

Now consider the case of floating-point numbers: 
as $\mathbb F$ is finite, \emph{every} decreasing sequence of valuations must necessarily stabilize with 
a fixed valuation.
\qed
\end{proof}
The termination result for the integers, as provided by Theorem \ref{t:term}, is remarkable.
It provides us with an effective technique also in case when the integer system is not
plainly polynomial --- but may use any kind of monotonic operators, in particular, integer division 
by constants.
%
Consider, e.g., the integer system \eqref{eqn:simple_loop2}
which requires integer division by 2. This seemingly innocent operation is \emph{not supported} 
in integer polynomial systems -- but is now covered by Theorem \ref{t:term}.

Termination of value iteration for floating-point numbers crucially relies on the monotonicity 
of the right-hand sides $e_{x,j}$ --- which in turn is guaranteed, e.g., if these are defined 
by affine arithmetic expressions with monotonic rounding.
%
%
If during \emph{max}-policy iteration, the value for a \emph{max}-policy is obtained by calling 
an optimization solver implemented in floating-point, inconsistencies can occur, leading to not necessarily
increasing sequences of valuations $\phi_t$. 
As a consequence, non-termination may be observed due to looping on \emph{max}-policies. 

Theorem \ref{t:term}, while providing termination guarantees, gives no hints on
how \emph{efficiently} the solution is computed. For the floating-point domain, 
we will come back to this question in Section~\ref{sec:floating-point}.


\section{\emph{Max-min} policy iteration}\label{sec:max-min}

In Section~\ref{sec:max}, we have presented the algorithm for \emph{max}-policy 
iteration --- but left open how the greatest solutions of the encountered \emph{max}-iterates
could be computed. 
As one technique, we presented decreasing value iteration for computing greatest solutions --- a technique, 
however, which may not terminate when solving systems over the rationals.
Consider, e.g., the equation
\begin{equation}
x = 1 \vee (0.5\cdot x + 3.0 \wedge 100.0) \vee (0.75 * x + 2.0 \wedge 100.0)  
\label{eq:rat-loop}
\end{equation}
The second encountered \emph{max}-policy $\pi_2 = \{x\mapsto 2\}$ results in the system
\begin{equation}
x = 0.5\cdot x + 3.0 \wedge 100.0
\label{eq:simple-rat-loop}
\end{equation}
Value iteration with exact arithmetic for the greatest solution yields a decreasing
sequence, starting from $100.0$ and converging to $6.0$ without ever reaching it.

\smallskip

When the system is bounded,
finding the greatest solution can be cast as an \emph{optimization problem} where the \emph{sum} 
of values of all unknowns is minimized \cite{GawlitzaS2011}.
%
Due to the surrounding \emph{max}-policy iteration,
we are guaranteed that each encountered optimization problem is \emph{bounded}. 
Thus, the minimum is attained at some 
finite point --- whose coordinates represent the desired solution.

\smallskip

Here, we introduce \emph{min}-policy iteration 
as an alternative approach, 
to determine the greatest solution of \eqref{eq:simple-rat-loop}. 
In this way, potentially expensive optimization problems are avoided
whose sizes grow with the size of the program.

In the given example, we start with the upper bound $100.0$ and then switch the choice 
at the minimum ---
to be now left with solving the equation
\begin{equation}
x = 0.5\cdot x + 3.0 
\label{eq:rat-eq}
\end{equation}
This system has three solutions for $x$, namely, $-\infty,6.0,$ and $\infty$. 
So, the problem remains which solution to choose.
%
Formally, assume that 
the right-hand sides $e_{x,j}$ in \eqref{rhs:max} are of the form
\begin{equation}
e_{x,j} = \mbox{$\bigwedge_{i\in I_{x,j}} e_{x,j,i}$}
\label{rhs:max_min}
\end{equation}
--- whenever $j\neq 0$, where each expression $e_{x,j,i}$ is affine, i.e., built up from constants and
unknowns by means of addition and multiplication with positive constants. 
For simplicity, we can further assume that
$e_{x,j,i}$ is of the form 
\begin{equation}
\mbox{$\sum_{x'\in X'} a_{x'}\cdot x'+ b$}
\label{d:affine}
\end{equation}
where $X'\subseteq X$; $0 < a_{x'} <\infty$ for all $x'\in X'$; and
$-\infty < b<\infty$.
Given a \emph{max}-policy $\pi:X\to{\mathbb N}_0$ with $X^\pi = \{x\in X\mid\pi(x)\neq 0\}$, 
we define a \emph{min}-policy as a mapping
$\mu:X^\pi\to{\mathbb N}_0$ where now
$\mu(x)$ is contained in $I_{x,\pi(x)}$, i.e., is a possible choice of an argument expression 
at the minimum in the right-hand side of $x$ in the \emph{max}-iterate $\Sigma^\pi$.
For convenience, we assume that there is a fixed finite upper bound $-\infty < M < \infty$,
such that for all $x\in X^\pi$, $M$ occurs as the 0th choice,
in every minimum expression \eqref{rhs:max_min}, i.e., $0\in I_{x,j}$ with $e_{x,j,0} = M$ for $j>0$.
For a \emph{max}-policy $\pi$ and a \emph{min}-policy $\mu$, the 
\emph{max-min iterate} $\Sigma^{\pi,\mu}$ is given by
\begin{equation}
x = e_{x,\pi(x),\mu(x)}\qquad(x\in X^\pi)
\label{eq:max_min}
\end{equation}
%
%
We may attempt to solve the system $\Sigma^\pi$ 
via \emph{min}-policy iteration as provided by 
\emph{dual} \emph{max}-policy iteration. 
According to this approach, however, we would be required to compute the
\emph{least} solutions of the systems \eqref{eq:max_min}.
The latter, however, is \emph{not} correct: The \emph{least} solution of our example 
equation \eqref{eq:rat-eq} is ${-\infty}$, which is not the intended solution.
Also, the \emph{greatest} solution of \eqref{eq:rat-eq} need not be correct.
What we must exploit here is that \emph{min}-policy iteration
is embedded into a surrounding \emph{max}-policy iteration.
Let $\pi_{t}$ denote a policy attained during \emph{max}-policy iteration 
and $\rho_{t}$ the current solution of $\Sigma^{\pi_{t}}$.
Let $\pi_{t+1}$ denote the next \emph{max}-policy, and
%
consider the mapping $\rho'_{t+1}$ given by 
$\rho'_{t+1}(x) = \sem{e_x}\,\rho_t = \sem{e_{x,\pi_{t+1}(x)}}\,\rho_{t}$ for $x\in X$.
Our goal is to determine the 
\emph{least} solution $\rho_{t+1}$ of $\Sigma^{\pi_{t+1}}$ with $\rho'_{t+1}\leq\rho_{t+1}$. 

%
By monotonicity, $\rho'_{t+1}$ must also be bounded by the 
solution determined for each \emph{max-min} iterate $\Sigma^{\pi_{t+1},\mu}$ 
attained during \emph{min}-policy iteration for $\Sigma^{\pi_{t+1}}$. 
Therefore, we 
select for any encountered \emph{min}-policy $\mu$, $\rho_{\pi_{t+1},\mu}$ as the \emph{least} solution
$\underline\rho$ of $\Sigma^{\pi_{t+1},\mu}$ with $\rho'_{t+1}\leq\underline\rho$ 
(to be determined by \textbf{least}).
Technically, we organize \emph{min}-policy iteration as follows (see Figure \ref{fig:min}).
Let $X'$ denote the set of unknowns $x\in X$ with $\rho' = \rho'_{t+1} >-\infty$. 
\begin{itemize}
\item 	We start with the \emph{min}-policy $\mu_0 = \{x\mapsto 0\mid x\in X'\}$,
	and the assignment $\rho_{t,0}(x) = M$ 
	if $x\in X'$ and $\rho_{t,0}(x)=-\infty$ otherwise.
\item	Given a policy $\mu_s$, we determine the next policy for $x\in X'$
	by setting $\mu_{s+1} (x) = \mu_{s}(x)$ if  
	$\sem{e_{x,\pi_{t+1}(x),\mu_s(x)}}\,\rho_{t,s} = \rho_{t,s}(x)$ and otherwise,
	setting it to some $i\in I_{x,\pi_{t+1}(x)}$ where 
	$\sem{e_{x,\pi_{t+1}(x),i}}\,\rho_{t,s} =
	\bigwedge_{i'\in I_{x,\pi_{t+1}(x)}} \sem{e_{x,\pi_{t+1}(x),i'}}\,\rho_{t,s}$.
\item	In case that $\mu_{s+1}$ is different from $\mu_s$, the assignment $\rho_{t,s+1}$ 
	is determined as the \emph{least} solution $\rho_{t,s+1}$ of
	$\Sigma^{\pi_{t+1},\mu_{s+1}}$ with $\rho'\leq\rho_{t,s+1}$.
\item	Assume that \emph{min}-policy iteration terminates with a policy $\mu_S$. 
	Then we construct a solution $\rho_{t+1}$ of $\Sigma^{\pi_{t+1}}$ by defining
	$\rho_{t+1}(x) = \rho_{t,S}$.
\end{itemize}
\begin{figure}[t]
\begin{center}
\[
\begin{array}{l}
\textsf{min\_policy\_iterator}(\Sigma^{\pi'},\rho')\;\{	\\
\qquad
\begin{array}[t]{l}
X'\;{:=}\;\{x\in X\mid\rho'(x)\neq -\infty\};i					\\
\mu\;{:=}\;\{x\mapsto 0\mid x\in X'\};						\\
\rho\;{:=}\;\{x\mapsto M\mid x\in X'\}\cup\{x\mapsto-\infty\mid x\not\in X'\}; 	\\
\textbf{do}\;\{ \\
\qquad\rho_1\;{:=}\;\{x\mapsto\sem{e_{x,\pi'(x),\mu(x)}}\,\rho\mid x\in X'\}\cup\{x\mapsto-\infty\mid x\not\in X'\};\\
\qquad\textit{modified}\;{:=}\;\textbf{false};  \\
\qquad\textbf{forall}\;(x\in X') \\
\qquad\qquad\textbf{if}\;(\rho_1(x) \neq \rho(x))\;\{ \\
\qquad\qquad\qquad\textit{modified}\;{:=}\;\textbf{true};       \\
\qquad\qquad\qquad\mu(x)\;{:=}\;\textbf{argmax}_{I_{x,\pi'(x)}}\{i\mapsto \sem{e_{x,\pi'(x),i}}\,\rho_1\};        \\
\qquad\rho = \textbf{least}(\Sigma^{\pi',\mu},\rho');   \\
\qquad\qquad\}  \\
\}\;\textbf{while}\,(\textit{modified});\\
\textbf{return}\;\rho;	\\
\}
\end{array}
\end{array}
\]
\caption{\label{fig:min}The \emph{min}-policy iteration algorithm for a new 
\emph{max}-policy $\pi'$ where $\rho'(x) = \sem{e_x}\,\rho = \sem{e_{x,\pi'(x)}}\,\rho$
for the previous solution $\rho$.}
\end{center}
\end{figure}
For this algorithm, we find:

\begin{lemma}\label{l:min}
\begin{enumerate}
\item	\emph{Min}-policy iteration for the system $\Sigma^{\pi_{t+1}}$
	terminates with the \emph{least} solution
	of $\Sigma^{\pi_{t+1}}$ subsuming $\rho'_{t+1}$.
\item	If $\Sigma$ has only concave \emph{max}-iterates, 
	$\rho_{t+1}$ equals the \emph{greatest} solution of $\Sigma^{\pi_{k+1}}$.
	%
	%
	In that case, the outer \emph{max}-policy iteration terminates as well.
	\qed
\end{enumerate}
\end{lemma}

\noindent
Let us call the new algorithm \emph{max-min}-policy iteration.
As systems of bounded affine rational equations have only concave \emph{max}-iterates,
we obtain:

\begin{theorem}\label{t:minmax}
\emph{Max-min} policy iteration for a system $\Sigma$ of bounded affine rational equations terminates 
and computes the least solution of $\Sigma$.
\qed
\end{theorem}

\noindent
Coming back to our example equation \eqref{eq:rat-loop} whose least solution $\rho^*$ 
is $8.0$. Consider the \emph{max}-policy $\pi_2$ with \eqref{eq:simple-rat-loop} 
as \emph{max}-iterate $\Sigma^{\pi_2}$. 
The greatest solution of that system is
$\rho_2 = \{x\mapsto 6.0\}$.  --- which equals the \emph{finite} solution 
$\rho^{\pi_2,\mu_1}$ of system \eqref{eq:rat-eq}.
The outer \emph{max}-policy iteration thus improves policy $\pi_2$ to 
$\pi_3 = \{x\mapsto 3\}$ 
--- for which ultimately
the solution $\{x\mapsto 8.0\}$ found, and both the inner \emph{min}-
and the outer \emph{max}-policy iteration terminate.


%
Theorem \ref{t:minmax} still leaves open \emph{how} the encountered systems 
\eqref{eq:max_min}
without maximum and minimum operators are solved.
Assuming that the set of unknowns $x$ with $\pi_{t+1}(x)>0$ is given by $\{x_1,\ldots,x_n\}$, 
that  system takes the form
\begin{equation}
x = A\cdot x + b
\label{eq:affine}
\end{equation}
for $x = [x_1,\ldots,x_n]^\top$ the column vector of unknowns, $b\in{\mathbb Q}^n$ and 
$A$ a square matrix ${\mathbb Q}^{n\times n}$ with non-negative entries.
Since this system of equations is guaranteed to have a \emph{unique} finite solution 
\cite{GawlitzaS2011},
%
$I_n-A$ must be invertible ($I_n$ the square identity matrix of dimension $n$). 
Therefore, the unique finite solution $\rho$ written as the vector
$[\rho\,x_1,\ldots,\rho\,x_n]^\top$, is given by
\begin{equation}
\rho = (I_n-A)^{-1}\cdot b
\label{eq:solution}
\end{equation}
For computing this solution, we may use Gaussian elimination or any other algorithm 
for solving linear systems of equations.  In summary, we thus obtain:

\begin{theorem}\label{t:bounded-rational}
\emph{Max-min} policy iteration reduces the problem of computing
the least solution of a system of bounded affine rational equations to solving a series 
of linear systems of equations over the rationals.
\qed
\end{theorem}

\section{Dealing with Infinity}\label{sec:infinity}

During \emph{max}-policy iteration, each value may move once from $-\infty$, meaning that a new location is deemed 
abstractly reachable. 
Let us see now how to deal with $+\infty$ values in affine rational systems. 
Here, we generally assume 
that $\infty + c = c+\infty = \infty$ for every value $-\infty < c \leq \infty$,
and  $c\cdot\infty = \infty$ for every value $0 < c \leq \infty$, while $0\cdot\infty = 0$.
In the following, we thus no longer assume that
all right-hand sides are bounded by finite upper bounds.
%
As a minimal example of why infinities create difficulties, consider the single equation 
\begin{equation}
x = {-1.0}\vee (0.5\cdot x + 2.0 \wedge 2.0\cdot x + 3.0 )
\label{ex:rat-infinity-i}
\end{equation}
\emph{Max}-policy iteration will at some point choose the second argument of $\vee$  
and consider the \emph{max}-iterate
\begin{equation}
x = 0.5\cdot x + 2.0 \wedge 2.0\cdot x + 3.0
\label{ex:rat-infinity-ii}
\end{equation}
Using as initial value $\infty$ for $x$, none of the two alternatives will be able to shrink this value --- implying that $\infty$ is returned.
Still, this system  has a finite least solution, namely, $\{x\mapsto {4.0}\}$.
The situation is different if the minimum contains any sufficiently large, but finite upper bound $M$:
\begin{equation}
x = M \wedge 0.5\cdot x + 2.0 \wedge 2.0\cdot x + 3.0 
\label{ex:rat-infinity-iii}
\end{equation}
Starting with the value $M$ for $x$, the choice $0.5\cdot x + 2.0$ is preferable 
for which the finite solution indeed provides us with the value $4.0$.
%
Our approach of dealing with this issue proceeds in two steps:
\begin{enumerate}
  \item Instead of solving the system for a \emph{specific} (sufficiently large) upper bound,
  	we introduce a \emph{symbolic} upper bound. This means that we consider a linear order whose
	elements represent numbers $M \eta + \beta$ ($M$ now a symbolic upper bound, $0\leq\eta<\infty$ and
	$-\infty\leq\beta < \infty$) and perform \emph{max-min} policy iteration for that.
  \item From the resulting solution, we extract a solution of the system over $\RR$ by
  	substituting $M$ with $\infty$. What we obtain in this wa, happens to be not any, but
	the least solution (Theorem \ref{t:limit}).
	%
\end{enumerate}
Recall that $\RR$ denotes the linear order of the reals, extended with $-\infty$ and $\infty$ as least and 
greatest elements.
First, we introduce an affine rational system $\Sigma$ of equations with \emph{variable} upper bound $z$ 
($z\not\in X$ a fresh variable, technically serving as a \emph{parameter}) of the form:
\begin{equation}
x = \mbox{$\bigvee_{j\in J_x}\bigwedge_{i\in I_{x,j}} e_{x,i,j}$}
\label{eq:symbolic}
\end{equation}
without occurrences of $\infty$, but 
where $0\in I_{x,j}$ for all $x\in X$ and $j\in J_x$, $j>0$ with $e_{x,j,0} = z$.
For every $-\infty < M$, we obtain from $\Sigma$ the system $\Sigma_M$ with upper bound $M$
by substituting $z$ with $M$. In case that $M = \infty$, we thus obtain the \emph{unbounded} system
$\Sigma_\infty$.

To deal with symbolic upper bounds, we introduce the linear order $\RR^*$ consisting of $-\infty$
together with all pairs $(\eta,\beta)\in\RR^2$ where $0\leq\eta<\infty$ and $-\infty <\beta<\infty$.
%
Intuitively, a pair $(\eta,\beta)$ represents the linear combination $\eta\cdot M+\beta$ for sufficiently large $M$.
The (strict) ordering $<$ on $\RR^*$ therefore, is given by
$-\infty < (\eta,\beta)$ and $(\eta,\beta) < (\eta',\beta')$
if $ \eta < \eta'$ or $\eta = \eta'$ and $\beta < \beta'$.
Note that this ordering is \emph{not} a complete lattice.
Consider, e.g., the sequence $(0,\beta_i),i\geq 0$, where $\beta_0 = 0$, and for $i>0$, 
$\beta_i = \beta_{i-1} +1$. Then the set upper bounds consists of all $(\eta,\beta)$ where 
$\eta > 0$, but there is no \emph{least} upper bound.

We extend maximum, minimum, and the arithmetic operators from $\RR$ to $\RR^*$ where
maximum and minimum are derived from the ordering on $\RR^*$, and
\[
\begin{array}{l@{\qquad}l}
c\cdot(\eta,\beta) = (c\cdot\eta,c\cdot\beta) &
(\eta,\beta) + (\eta',\beta') = (\eta+\eta',\beta+\beta')	\\
\multicolumn{2}{l}{c\cdot(-\infty) = -\infty + r = r + (-\infty) = -\infty}
\end{array}
\]
for $0\leq c <\infty$, and $r,(\eta,\beta),(\eta',\beta')\in\RR^*$.
For every expression $e$ without occurrences of $\infty$, we provide an evaluation function
$\sem{e}^*:(X\to\RR^*)\to\RR^*$ by
\[
\begin{array}[t]{lll}
\sem{-\infty}^*\,\rho	&=& -\infty		\\
\sem{b}^*\,\rho	&=& (0,b)\qquad (-\infty < b < \infty) \\
\sem{x}^*\,\rho	&=& \rho(x)\qquad (x\in X) 	\\
\end{array}\qquad
\begin{array}[t]{lll}
\sem{z}^*\,\rho &=& (1,0)			\\
\sem{c\cdot e}^*\,\rho	&=& c\cdot\sem{e}^*\,\rho	\\
\sem{e\,\Box\,e'}^*\,\rho	&=& (\sem{e}^*\,\rho)\,\Box\,(\sem{e}^*\,\rho)	\\
\end{array}
\]
for all binary operators $\Box\in\{+,\wedge,\vee\}$. 
From the evaluation of an expression $e$ over $\RR^*$, we may recover the corresponding value over $\RR$
w.r.t.\ some upper bound $-\infty < M\leq\infty$, by means of the mapping $\alpha_M: {\RR}^*\to\RR$ with
\[
\begin{array}{lll@{\qquad}lll}
\alpha_M(-\infty)	&=& {-\infty}	&
\alpha_M(\eta,\beta) &=&\eta\cdot M+\beta	 
\end{array}
\]
\begin{lemma}\label{l:hom}
For every $-\infty < M$,
$\alpha_M$ is monotonic and commutes with $+$, scalar multiplication, minimum and maximum,
and we have for every $r\in{\RR}^*$,
\begin{enumerate}
\item	If $r < r'$ in $\RR^*$, then there is some $-\infty < M <\infty$ such that
	for all $M\leq M' < \infty$,
	$\alpha_{M'}(r) < \alpha_{M'}(r')$;
\item	$\alpha_\infty(r) = \bigvee_{M\leq M' < \infty} \alpha_{M'}(r)$ for all $M\geq 0$.
	\qed
\end{enumerate}
\end{lemma}

\noindent
\emph{Max-min} policy iteration 
for the equation system $\Sigma$ given by \eqref{eq:symbolic}
can also be performed over the linear order $\RR^*$.
The encountered \emph{max-min} iterates 
still are of the form \eqref{eq:affine} --- where $A$ is a square matrix with $I_n-A$ being invertible.
Only the vector $b$ now has entries in $\RR^*$. 
Thus, this system has the valuation $\rho$ given by \eqref{eq:solution} as its unique finite solution.
%
Moreover, due to Lemma \ref{l:hom}, $\alpha_M\circ\rho$ is the solution of the system 
$x = A\cdot x+ b_M$ where $b_M$ is obtained from $b$ by applying $\alpha_M$ to each of its entries.

Let $\Pi = \pi_0,\ldots,\pi_N$ denote the sequence of encountered \emph{max}-policies, 
and for each $t=0,\ldots,N$, the corresponding sequence $\mu_{t,0},\ldots,\mu_{t,M_t}$ 
of encountered \emph{min}-policies where for each pair $(t,s+1)$, $\rho_{t,s+1}$ 
is the valuation for the \emph{max-min} iterates $\Sigma^{\pi_{t+1},\mu_{s+1}}$
encountered during the \emph{min}-policy iteration for the system.
Let $\rho_{t+1} = \rho_{t,M_t}$.
%
For any $-\infty < M\in\RR$, let $\Sigma_M$ denote the system where the variable upper bound $z$ in 
$\Sigma$ is interpreted as $M$.

\begin{theorem}\label{t:limit}
There is some $-\infty < M < \infty$ such that for every $M \leq M' <\infty$, 
\begin{enumerate}
\item	\emph{max-min} policy iteration for $\Sigma_{M'}$
	traverses the same sequence of \emph{max}-policies, and for each $t\in\{0,\ldots,N\}$,
	the same sequence of \emph{min}-policies where for 
	all $s\geq 0$,
	$\alpha_{M'}\circ\rho_{t,s}$ is the least solution of $\Sigma_{M'}^{\pi_{t+1},\mu_s}$ subsuming
	$\alpha_{M'}\circ\rho_{t}$.
\item	$\alpha_{M'}\circ\rho_N$   
	is the least solution of $\Sigma_{M'}$.
\item	$\alpha_\infty\circ\rho_N$ 
	is the least solution of $\Sigma_\infty$.
\end{enumerate}
\end{theorem}

\begin{proof}
The first two statements of the theorem follow from Lemma \ref{l:hom}.
Now consider statement (3).
For each $M'\geq M$, we construct a sequence of valuations $\sigma_{M',j}, j\geq 0$,
defined by $\sigma_{M',0}= \{x\mapsto-\infty\mid x\in X\}$, and for $j>0$,
\[
\sigma_{M',j} = \{x\mapsto\sem{e'_{x}}\,(\sigma_{M',j-1}\oplus\{z\mapsto M'\})\mid x\in X\})
\]
where the operator ''$\oplus$'' is meant to
extend the mapping to the left with another argument-value pair.
For each $M'$, the sequence $\sigma_{M',j},j\geq 0,$ is ascending
where $\bigvee_{j\geq 0} \sigma_{M',j}$ equals the least solution of $\Sigma_{M'}$.
We claim that for all $j\geq 0$,
\[
\sigma_{\infty,j} = \mbox{$\bigvee_{M\leq M' < \infty}\sigma_{M',j}$}
\]
\ignore{
For a proof of this claim, we exploit that
all operators occurring in $\Sigma$ are \emph{Scott-continuous}, i.e.,
commute with least upper bounds of \emph{directed} subsets\footnote{
A subset $D$ is \emph{directed}, if for each $a,b\in D$, there is some $c\in D$ with
$a\leq c$ and $b\leq c$.}.
Since  projections and constant functions are Scott-continuous, 
and Scott-continuous functions are closed under composition,
expression evaluation for right-hand sides is Scott-continuous as well.
}
This claim is proven by induction on $j$.
For the base case $j=0$,
we have that $\sigma_{M',j}(x) = -\infty = \sigma_{\infty,j}(x)$ for all $x\in X$, and the assertion holds.
Now assume that the claim already holds for $j$. Then for every $x\in X$,
\[
\begin{array}{lll}
\bigvee_{M\leq M' < \infty}\sigma_{M',{j+1}} (x) 
&=& 
\bigvee_{M\leq M' < \infty}\sem{e_x}(\sigma_{M',j}\oplus\{z\mapsto M'\})	\\
&=&
\sem{e_x}(\bigvee_{M\leq M' < \infty}(\sigma_{M',j}\oplus\{z\mapsto M'\})) \quad\text{by Scott-continuity}	\\
&=&
\sem{e_x}(\sigma_{\infty,j}\oplus\{z\mapsto\infty\})) \quad\text{by induction hypothesis}	\\
&=&
\sigma_{\infty,j+1}(x)
\end{array}
\]
and the claim follows. By means of the claim, we verify that
\[
\begin{array}{lll@{\quad}l}
\bigvee_{j\geq 0} \sigma_{\infty,j} &=&
	\bigvee_{j\geq 0}\bigvee_{M\leq M'<\infty}\sigma_{M',j}		\\
&=&
	\bigvee_{M\leq M'<\infty} \bigvee_{j\geq 0} \sigma_{M',j} 	\\
&=&
	\bigvee_{M\leq M'<\infty} \alpha_{M'}\circ\rho_N 	&(*) 	\\
&=&	\alpha_\infty\circ\rho_N				&(**)
\end{array}
\]
Here, $(*)$ follows by statement (2) and $(**)$ by Lemma \ref{l:hom} (3).
\qed
\end{proof}
	
\noindent
By Theorem \ref{t:limit}, we thus have effectively reduced solving systems of unbounded affine rational equations 
to \emph{max-min} policy iteration for affine rational equations with symbolic upper bounds over $\RR^*$.

\section{Witnesses}\label{sec:witness}

\ignore{
We propose two extensions to our approach for solving max-min systems: 
one that produces \emph{witnesses} that a solution is correct and/or optimal, and 
another for analyzing systems with parameters. 
%
%
In fact, the \emph{lazy approach} for parameterized systems will be based on 
witnesses of optimality. 
Parametric analysis can be used to implement forms of modular analysis.

\subsection{Witnesses of Solution of Max-Min Systems}
\label{sec:witness}
}

The search for the least solution of a system of affine rational equations, like in general finding
least inductive invariants in complicated abstract domains, is an expensive and complex task. 
We ask how this process can log a \emph{witness} that justifies that the result of the process 
is \emph{correct} (it is a solution) and \emph{optimal} (it is the least one). 
A witness $W$ for valuation $\rho$ should provide enough information that a 
\emph{verifier} can easily check that $\rho$ has the desired property.
For the construction of witnesses and their verifiers, we try to capture three ideas: 
\begin{itemize}
\item	The verifier should have \emph{limited} computational power (time and memory); 
\item 	it should be \emph{dissimilar} to the solver; and
\item	it should be \emph{simple} enough so that the implementation is easy to review or even 
	to formally be proven correct. 
\end{itemize}
Limited computational power means that a run by the verifier should be significantly less expensive 
than solving the problem from scratch. We make this idea precise through the complexity class \NCone.
{\NCone} is the class of problems decidable by LOGSPACE-uniform Boolean circuits with a polynomial number 
of gates of at most two inputs and depth $O(\log n)$, or 
decidable in time 
$O(\log n)$ on a parallel computer with a polynomial number of processors.

An extreme way of constructing witnesses is to log a full execution trace of the solver.
Verification of the witness then boils down to verifying that each step of the solver has been 
correctly executed -- which can be checked in \NCone. 
The hidden assumption, however, is that the solver itself was correct, as any bug in the solver 
would go unnoticed --- dooming both the produced witness and its verification as incorrect (in some cases, at least).
Witness construction, therefore, should omit as many details of the solver as possible.

Let us again consider a rational affine system (not necessarily bounded),
and a valuation $\rho$.
For the valuation $\rho$ to be \emph{correct}, i.e., a solution of the system,
it is sufficient for each unknown $x$, to compare the value $\rho(x)$ with the 
value returned by the right-hand side of $x$ for $\rho$.
Expression evaluation itself can be realized within \NCone,
no further witness thus is required.

The situation for optimality of $\rho$ is more intricate.
Our proposed witness format for optimality records a \emph{footprint} of the solver's computation so that the checker 
still can \emph{easily} fill in the remaining parts. 
For the construction of this footprint, we go back to the corresponding system $\Sigma$
with symbolic upper bounds and record
the sequence $\pi_0,\ldots,\pi_N$ of \emph{max}-policies encountered by \emph{max}-policy iteration for $\Sigma$.
%
For each encountered \emph{max}-policy $\pi_t$, 
we furthermore record the finite solution $\rho_t$ over $\RR^*$
of the respective \emph{max}-iterate.

Let $W = (\pi_0,\rho_0)\ldots(\pi_N,\rho_N)$ denote a witness.
The algorithm checking $W$ is a witness for optimality of the valuation $\rho$ works at follows:
\begin{enumerate}
\item	It checks that $\pi_0$ is the initial policy, and that 
	$\rho_N$ is a solution of the given system with symbolic upper bounds
	where $\rho = \alpha_\infty\circ\rho_N$;
\item 	For each $t=0,\ldots,N$, it checks that $\rho_t$, 
	is a solution of the system $\Sigma^{\pi_t}$;
\item 	For each $t=0,\ldots,N-1$, it checks that $\pi_{t+1}$ is a
	reluctant improvement of $\pi_t$ (given the valuation $\rho_t$).
\end{enumerate}
For that certification to work, we rely on the observation that over $\RR^*$, 
each $\Sigma^{\pi_t}$ has a unique finite solution since for every finite bound $M$, there is one.
In this way, we have reduced checking of optimality to a sequence of checks of candidate 
solutions together with a sequence of local checks that policies have correctly been improved.
What we could purge completely is \emph{how} the \emph{max}-iterates $\Sigma^{\pi_{t+1}}$ 
have been solved.

Each of the tasks (1), (2) and (3) can be performed in \NCone.
Optimality of $\rho$ then follows from Theorem \ref{t:limit}. Therefore, we obtain:
\begin{theorem}\label{t:witness}
  The given witness of optimality of a valuation for a system of affine rational equations can be 
  verified in \NCone.
\qed
\end{theorem}

\smallskip
\noindent
Our witness format for optimality seems cumbersome, 
as it requires the enumeration of encountered \emph{max}-policies.
Here, we give a complexity-theoretical explanation why it is unlikely that much shorter witnesses can be used.
Consider a restricted version 
of the problem, namely, minimal solutions to systems of \emph{propositional Horn clauses}.
A propositional Horn clause is a disjunction of literals of which at most one has a positive occurrence. 
A propositional Horn clause that contains no positive literal is deemed \emph{purely negative}. 
A non-purely negative Horn clause (NPNHC) is thus of the form $\bar{A}_1 \lor \dots \lor \bar{A}_n \lor C$ 
($A_1,\dots,A_n,C$ are variables, $n$ can be null). Such a clause is denoted as 
$C\leftarrow A_1 \land \dots \land A_n$, 
where $C$ is called the \emph{consequent}.
Grouping all non-purely negative Horn clauses with the same consequent $C$, 
we obtain a system of implications of the form 
$L\leftarrow \big((A_{1,1} \land \dots \land A_{1,n}) \lor \dots \lor (A_{m,1} \land \dots \land A_{m,n})\big)$
with distinct $C$ ($m$ can be null when no clause has $C$ as a consequent).
By casting the Boolean values \textbf{false} and \textbf{true} as $-\infty$ and some finite constant $M$, respectively,
and logical conjunction and disjunction as minimum and maximum, respectively, we can consider each such
system as a system of equations over $\{-\infty,M\}$.
A \emph{solution} of such a system is thus a post-fixpoint of an operator mapping a valuation of the 
variables to the corresponding values of the left-hand sides of these implications.
\ignore{
By Tarski's fixed point theorem, this system has a minimal solution (for the product ordering arising from 
$-\infty < \infty$), 
which is also the minimal solution of a system of the form 
$C = \big((A_{1,1} \land \dots \land A_{1,n}) \lor \dots (A_{m,1} \land \dots \land A_{m,n})\big)$ over the reals.

We said in Section~\ref{sec:witness-correct} that it is easy to check that a valuation is truly a solution. In the case of propositional Horn clauses, this translates into this complexity-theoretical statement:

\begin{lemma}\label{lem:inductiveness-in-nc}
  Given a system of NPNHC $H$ and a valuation $\rho$, checking that $\rho$ is a solution of $H$ is in \NCone.
\end{lemma}

\begin{proof}
  The left-hand side of each Horn clause $A_1 \land \dots \land A_n \implies C$ can be seen as a binary tree of $\land$ of logarithmic depth.
  Apply this process in parallel to all clauses and check that each implication holds.
\end{proof}
}
While checking that a valuation from unknowns to $\{-\infty,M\}$ is a solution, is in \NCone,
checking that a valuation is the \emph{least} solution is \compclass{P}-complete:%
\footnote{We thank Emil Jeřábek for this proof. Note that if the right-hand sides contain only $\vee$ but no $\wedge$,
the problem becomes a reachability problem in an explicitly represented transition system, an equation 
$d = s_1 \vee \dots \vee s_n$ meaning $n$ transitions $s_i \rightarrow d$. 
This problem, known as STCON, instead is \compclass{NL}-complete.}

\begin{theorem}\label{th:reachable-p-complete}
Given a system of NPNHC $H$ and a valuation $\rho$,
checking that $\rho$ is the \emph{least} solution of $H$ is \compclass{P}-complete 
(with respect to log-space reductions).
\end{theorem}

\begin{proof}
  The proof is by reduction from the problem of whether a system of Horn clauses $H$ can deduce the empty clause. 
  The latter problem is known to be \compclass{P}-complete
  (by reduction from circuit value \cite[p.~213]{Cook_Nguyen}).
  We add to $H$ an extra literal $f$, and 
  change each purely negative clause $\bar{A}_1 \lor \dots \lor \bar{A}_n$ into 
  $\bar{A}_1 \lor \dots \lor \bar{A}_n\lor f$. In this way, we obtain a system $H'$
  consisting of a conjunction of non-purely negative Horn clauses only.
  Then $f$ is mapped to \textbf{true} (i.e., to $M$) by the least solution of $H'$ if and only if 
  $H$ has no solution (which is equivalent to being able to derive the empty clause).
  By adding to $H'$ clauses $\bar{f} \lor A$ for all literals $A$, we obtain a set $H''$
  of Horn clauses so that 
  the minimal model of $H''$ assigns \textbf{true} to all literals iff $H$ is unsatisfiable,
  i.e., can deduce the empty clause.
  \qed
\end{proof}
Under the usual conjecture that 
\NCone is strictly included in \compclass{P}, 
we thus cannot hope to verify in {\NCone} that a valuation is the least solution of a system of equations.

\ignore{
Let us now go back to linear numerical systems (not just Booleans).  A fortiori, we cannot verify in \NCone, in the more general numerical case, that a valuation is truly a minimal solution of a min-max affine linear equation system.
However, the certificates formats that we use (which, in the case of minimality, embark a trace, not just the solution), can be checked in \NCone.
}

\section{Solving in Floating-Point, for Floating-Point}
\label{sec:floating-point}

\emph{Max-min} policy iteration for systems of (bounded or unbounded) affine equations assumed exact (rational) 
arithmetic. In this section, we give several approaches for 
solving systems using floating-point arithmetic.
%
There are two reasons why this is important.
First, solving many subproblems in exact precision arithmetic may be expensive, 
even despite the use of efficient algorithms that eschew Gaussian elimination in arbitrary precision 
rational arithmetic \cite{Flint,dumas:hal-02102080,Dixon1982}.
Second, often the programs under analysis use floating-point semantics for arithmetic.
In this case, one can abstract it into a program over the rational numbers, liable to analysis 
by a tool itself implemented using multiprecision rational numbers, as follows \cite[\S7.4]{Miné_PhD}:%
\footnote{Here, we assume that overflows and the generation of $\pm\infty$
are handled by the insertion of appropriate \emph{if-then-else} statements.}
a floating-point addition $x \,+^\sharp\, y$ yields $x + y + \epsilon_{x+y}$ where $\epsilon_{x+y}$ is chosen nondeterministically so that $|\epsilon_{x+y}| \leq \epsilon_r |x + y|$ where $\epsilon_r$ is the \emph{relative error} coefficient;
a floating-point multiplication $x \,*^\sharp\, y$ yields $x \cdot y + \epsilon_{x \cdot y}$ where $\epsilon_{x \cdot y}$ is chosen nondeterministically so that $|\epsilon_{x \cdot y}| \leq \epsilon_r |x \cdot y| + \epsilon_a$ where $\epsilon_a$ is the \emph{absolute error} coefficient.%
\footnote{The absolute error is due to underflow, that is, rounding to 0 of values close to 0. In IEEE-754 floating-point, one can show that there is no absolute error on addition and subtraction.}
This abstraction however may cause scalability issues:
$\epsilon_r$ and $\epsilon_a$ have large denominators, and this will tend to generate huge numerators and 
denominators in the rational analysis.

We have investigated three approaches based on \emph{max}-policy iteration
which instead directly work with floating point arithmetic.
The three approaches differ in how they deal with encountered \emph{max}-iterates 
$\Sigma^\pi$ ($\pi$ some \emph{max}-policy):
\begin{description}
\item[\textbf{Max+LP}] solving $\Sigma^\pi$ by calling an external linear programming solver;
\item[\textbf{Max+Min}] solving $\Sigma^\pi$ by \emph{min}-policy iteration using Gaussian elimination;
\item[\textbf{Max+Val}] solving $\Sigma^\pi$ by descending value iteration.
\end{description}
The method based on \emph{max-min} policy iteration may also be combined with symbolic upper bounds.
To avoid extra complications due to infinities for LP solving, we concentrated on \emph{bounded} systems.
Subsequently in this section, 
we therefore consider \emph{finite} values only, i.e., values $v\in\RR$ with $v <\infty$.

The natural approach for solving systems with right-hand sides \eqref{rhs:max_min}
and expressions $e_{x,j,i}$ built up by addition and multiplication with positive constants,
is to replace each equality $y = x_1 \wedge \dots \wedge x_n$ by inequalities 
$y \leq x_1$, \dots, $y \leq x_n$ and to maximize for all such $y$, 
which itself is equivalent to maximizing 
the \emph{sum} of all variables \cite[Section~2]{DBLP:conf/csl/GawlitzaS07}.
Such a problem can be solved by a linear programming (LP) solver.

Most off-the-shelf solvers are implemented with floating-point arithmetic.
%
%
There is no guarantee that the results will satisfy the hypotheses needed for ascending policy iteration.
During some of our experiments, \textbf{Max+LP} ended up cycling between policies. 
Interestingly, different solvers may even behave differently here (see Section~\ref{sec:experiments}),
and there is no guarantee that the computed valuation is a solution, i.e., provides us with an inductive invariant 
for the floating-point program under analysis.
In contrast, non-termination was neither observed for
\textbf{Max+Min} nor for \textbf{Max+Val}.

Regarding \textbf{Max+Min},
Gaussian elimination over floating-point numbers may, at least in principle, result in imprecision and 
be computationally unstable.
Therefore, we experimented with two-stage approaches (\textbf{Max+LP+Val}, \textbf{Max+Min+Val}) where the first computes approximate invariants
which are then in the second stage, \emph{repaired} 
by running \emph{value iteration} until a fixpoint was reached.%
\ignore{
%
%
Note that this approach differs from that \citet{Miné_PhD,DBLP:journals/lisp/Mine06} 
applied to the octagon abstract domain. He implemented abstract transformers in floating-point 
using rounding to $+\infty$ to soundly over-approximate an ideal computation over the reals \cite[\S7.5.3]{Miné_PhD}.
One difficulty that he faced was that, due to rounding issues, standard widening operators could fail to induce convergence and he had to use ``perturbated'' versions that would ``fatten'' the invariant candidates to ensure convergence \cite[\S7.5.4]{Miné_PhD}. Instead of such a ``fattening'' step we use a limited number of value iteration steps.
}
Experiments showed that the overhead incurred by these additional iterations was mostly negligible.
In fact, for \textbf{Max+Min+Val}, subsequent value iterations were never required.
This was different in \textbf{Max+LP+Val} --- although there 
a limited number of value iterations sufficed, except in a few exceptional cases.
In these cases, the floating-point LP solver had produced a negative value $-\epsilon$ close to $0$, 
whereas the \emph{true} solution was $0$. Then value iteration would compute an \emph{ascending} sequence 
converging to $0$ with a sequence of floating point values $m_k2^{-e_k}$ with $-e_k$ going through the negative 
exponents up to the minimum exponent $-1022$, and then some denormal numbers and only then flush to zero. 
This used $1000$ iterations for no good reason.
In contrast, convergence to nonzero values went through a number of iterations seemingly proportionate to the number 
of bits in the mantissa of the numbers ($53$).

We leave it to future work to formally prove 
geometric convergence of value iteration in \textbf{Max+Val} 
for solving the encountered \emph{max}-iterates.
%
For that, results from Perron-Frobenius theory \cite{Boyle_Perron_Frobenius} would have to be
extended from irreducible matrices to our more general situation.

\section{Experiments}\label{sec:experiments}

We have conducted experiments with the approaches \textbf{Max+LP}, \textbf{Max+Min}, and
\textbf{Max+Val} for solving bounded systems of affine floating-point equations. 
In order to generate many examples of arbitrary sizes, we use a random generator of analysis problems. This generator creates $n$ unknowns $x_i$, each defined (depending on a random choice) as an expression: either a random number in $[-1000,1000]$, or $x_j \lor x_k$, or $x_j \land x_k$, or $x_j + x_k$ where $x_j$ and $x_k$ are two randomly chosen unknowns, or $k \times x_j$ where $k$ is a randomly chosen constant in $[0,3]$ and $x_j$ a randomly chosen unknown, or $k_1 + x_j \lor k_2$ where $k_1$ and $k_2$ are constants in $[0,3]$ and $[-100,100]$, or $k_1 + x_i \lor k$ where $k$ is a constant in $[-100,100]$. The last case is meant to force the introduction of looping structures, as they happen in systems extracted from programs.

\begin{figure}
\includegraphics[width=0.5\textwidth]{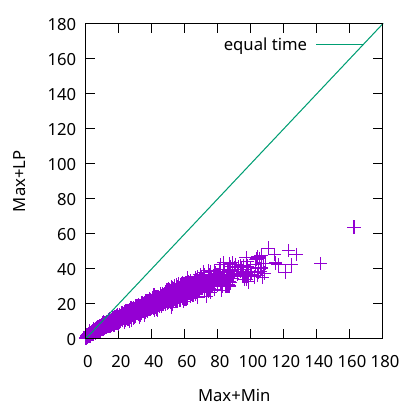}%
\includegraphics[width=0.5\textwidth]{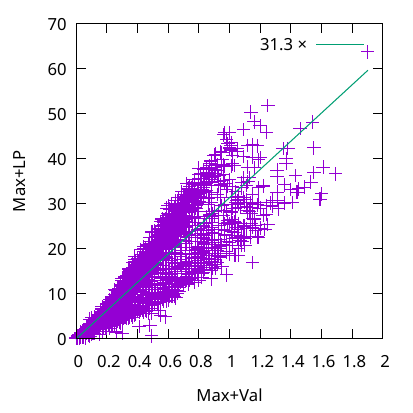}
\caption{Comparison of timings (seconds) for 3800 bounded examples with 
$[200,3999]$ unknowns for the ($x$ axis) \textbf{Max+Min} (left), or \textbf{Max+Val} (right), and \textbf{Max+LP} using Gurobi ($y$ axis). Roughly, \textbf{Max+Min} is $2\times$ slower than \textbf{Max+LP} (which may be explained by a naïve implementation), but \textbf{Max+Val} is $31 \times$ faster than \textbf{Max+LP}.}
\label{fig:times_minmax_lp}
\label{fig:times_maxval_lp}
\end{figure}

%
We implemented the outer \emph{max}-policy iteration, the inner \emph{min}-policy as well as value iteration in \textsc{Ocaml 5.3}.
For linear programming
we used two mature, off-the-shelf solvers, namely, GLPK 5.0 (\url{https://www.gnu.org/software/glpk/}), 
which is free software, and Gurobi 12.0 (\url{https://www.gurobi.com}), a proprietary, closed-source, high-performance, commercial package. 
Both were used through a common \textsc{OCaml} interface, LP 0.4 
(\url{https://ocaml.org/p/lp/0.4.0}).
The \emph{max-min} policy iteration algorithm 
requires for each encountered pair of a \emph{max}- and \emph{min}- policy to solve an affine system of
equations.
We solve these systems by Gaussian elimination, each from scratch.
An advanced implementation, however, may track dependencies between unknowns and thus perhaps
avoid recomputing unchanged values%
\footnote{In the same way that a LP solver does not recompute the full tableau when pivoting the basis.}.
Also, instead of using our own implementation, we could have relied on some external 
state-of-the-art solver for sparse systems.
Yet, despite this na\"ive implementation, the timings for \textbf{Max+Min} and \textbf{Max+LP}
differ only by a small factor (Fig.~\ref{fig:times_minmax_lp}), even though the LP solver 
is highly optimized.
Also, while we had no difficulty running \textbf{Max+Min}, we had some issues with \textbf{Max+LP}:
\ignore{
%
%
One is technical: proper detection of which variables go to 
$+\infty$ (by exhibiting, when the linear program is unbounded, a ray $o + \mathbb{R}_+ v$) 
requires information not available through the interface (one would have to query the last 
primal simplex base reached and, see which nonbasic variable(s) is allowed to grow indefinitely, 
and for each such variable get a vector $v$ as a column). We thus detect infinities when values 
cross a large bound ($200000$); this is just for purposes of experimentation.
}
%
we found examples with few variables where 
numerical round-off issues cause \textbf{Max+LP} to cycle between \emph{max}-policies. 
We implemented various measures based on $\epsilon$ thresholds to avoid that unappreciated non-termination. 
Gurobi was quite more reliable than GLPK in this respect.
In almost all experiments, the value iteration in \textbf{Max+LP+Val} or \textbf{Max+Min+Val} has barely discernible overhead.
We also compared these algorithms with \textbf{Max+Val} where value iteration (either in a round-robin fashion
or by a worklist) is used for solving systems without maximum operators.
As shown by Figure~\ref{fig:times_maxval_lp}, \textbf{Max+Val}
vastly outperforms (by a speed factor of about $30$) \textbf{Max+LP} 
-- even when using a state-of-the-art linear programming solver (Gurobi).

In our experiments, we checked that the different approaches yielded comparable results. 
In almost all cases the results are identical, 
and in few cases they differed by negligible amounts ($< 10^{-11}$).

\newcommand{\D}{{\mathbb D}}

\section{Extension: Template linear constraints}
%
%
\label{sec:extensions}
The systems of equations over rational or floating-point numbers which we have considered so far
support, besides maximum and minimum operators only multiplication with positive constants and addition.
A precise \emph{bounds} analysis, however, additionally requires the operator ``?'' 
from Section \ref{sec:background} as well as the \emph{sequence} operator ``;'' 
defined by $a\,;\,b = -\infty$ if $a= -\infty$ and $b$ otherwise.
Since these operators can be conveniently integrated into \emph{max}-policy iteration, we did not include these 
into our technical exposition.

%

Beyond \emph{bounds} analysis, we might want to identify \emph{relational} dependencies between program variables.
Assume that the set of program variables is given by $X= \{\textbf x_1,\ldots,\textbf x_n\}$ so that
a concrete program state can be represented by a vector $x\in\D^n$ where the value of $\textbf x_i$ is the
$i$th component $x_i$ of $x$.
Here, we briefly discuss the extension of our method, first, to the weakly relational domain of \emph{octagons}
\cite{DBLP:journals/lisp/Mine06} and then, 
to arbitrary template linear constraints domains \cite{DBLP:conf/vmcai/SankaranarayananCSM06}. 

\subsection{Octagon constraints}
\label{ss:octagons}
Octagons generalize intervals in that they not only allow to express ranges of
program variables, but also of the sum and difference of two program variables.
An element of the octagon domain ${\mathbb O}_\D$ thus either equals $\bot$ (signifying the property \textsf{false}) 
or is a satisfiable finite conjunction of inequalities of the form 
$\epsilon_i {\bf x}_i\leq c$ or $\epsilon_i{\bf x}_i + \epsilon_j{\bf x}_j\leq c$
for program variables ${\bf x}_i,{\bf x}_j$, $\epsilon_i,\epsilon_j\in\{1,-1\}$,  and upper bounds $c\in\D$. 
Let us call these inequalities \emph{octagon inequalities}, and their
left-hand sides \emph{octagon linear forms}. 
%
%

As indicated in \cite{DBLP:journals/lisp/Mine06,DBLP:conf/vmcai/BagnaraHZ08,DBLP:conf/sas/SchwarzS23}, each octagon $o$ (over the integers as well as over the rationals)
%
%
can be \emph{normalized} in cubic time into an equivalent octagon $o'$ which either equals $\bot$
or is a finite conjunction of octagon inequalities so that 
for every inequality $\ell\leq c$ in $o$,
$o'$ has an inequality $\ell\leq c'$ with $c'\leq c$.

\paragraph{Octagon normalization.}
Assume we are given a vector $c$ over $\D$ which is indexed by octagon linear forms
where $c[\ell]$ provides the minimum of upper bounds $c$ in octaton inequalities $\ell \leq c$ in $o$.
%
%
Conceptually, normalization can be computed in three stages.
The first stage considers two-variable inequalities only. 
A resulting upper bound $a[\epsilon_i {\bf x}_i+\epsilon_j {\bf x}_j]$ for the two-variable form
$\epsilon'_1 {\bf x}_i +\epsilon'_2 {\bf x}_j$ ($i\neq j$)
is the minimum over all sums
$
\mbox{$\sum_{k=0}^{m-1} c[-\epsilon_{k}{\bf x}_{i_k}+\epsilon_{{k+1}}{\bf x}_{i_{k+1}}]$}
$
where all $\epsilon_k{\bf x}_{i_k}$ are pairwise distinct, $i = i_0, j=i_m$, $\epsilon'_{1} = -\epsilon_{i_0}$
and $\epsilon'_2 = \epsilon_{m}$.
In the second phase, a vector $n$ of upper bounds for one-variable octagon linear forms is obtained by
\[
\begin{array}{lll}
b[\epsilon {\bf x}_i] &=& c[\epsilon {\bf x}_i] \wedge
		\bigwedge_{j\neq i}(a[\epsilon {\bf x}_i + {\bf x}_j] + a[\epsilon {\bf x}_i - {\bf x}_j])/2 \\
n[\epsilon {\bf x}_i]  &=& b[\epsilon {\bf x}_i] \wedge
			\bigwedge_{j\neq i} b[{\bf x}_j] + a[-x_j,\epsilon {\bf x}_i] \wedge
			\bigwedge_{j\neq i} b[-{\bf x}_j] + a[x_j,\epsilon {\bf x}_i]
\end{array}
\]
where division by 2 is interpreted over the integers 
or the rationals, respectively.
Finally, the upper bounds $n[\epsilon_i x_i+\epsilon_j x_j]$ for two-variable octagon linear forms
are possibly improved taking the one-variable upper bounds into account, i.e.,
\[
n[\epsilon_i {\bf x}_i+\epsilon_j {\bf x}_j] = 
a[\epsilon_i {\bf x}_i+\epsilon_j {\bf x}_j] \wedge
(n[\epsilon_i {\bf x}_i] + n[\epsilon_j {\bf x}_j])
\]
Assume that $o_1,o_2$ are normalized, different from $\bot$ and represented by the vectors $c_1,c_2$
respectively. 
Then the vectors $\bar c,\underline c$ representing the least upper bound $o_1\vee o_2$ and greatest lower
bound $o_1\wedge o_2$ are obtained as 
\[
\begin{array}{lll}
\bar c[\ell] &=& c_1[\ell] \vee c_2[\ell] 	\\
\underline 
c[\ell]      &=& c_1[\ell] \wedge c_2[\ell] \qquad(\ell\;\text{an octagon linear form})
\end{array}
\]
--- where, however, the octagon $o_1\wedge o_2$ must subsequently be normalized.
%
%

\paragraph{Octagon analysis.}
%
%
Assume that the program is given by a control-flow graph as in Section \ref{sec:background}
where each control-flow edge from program point $u$ to program point $v$ either is a 
guard or an assignment. As in \cite{DBLP:journals/lisp/Mine06}, we only consider guard expressions $g$ which themselves can be
expressed as an octagon. Then
\begin{equation}
\sem{g?}^\sharp\,o = \textsf{normalize}\,(o\wedge g)
\label{def:guard}
\end{equation}
Also  assignments are considered with unknown (or too complicated) right-hand sides
${\bf x}_i\;{:=}\;?$ (input) or of the form 
${\bf x}_i\;{:=}\,c$ or ${\bf x}_i\,{:=}\,\epsilon {\bf x}_i +c$ (invertible assignments)
or ${\bf x}_i\,{;=}\,\epsilon {\bf x}_j+c$ with $j\neq i$ (noninvertible assignments)
with $\epsilon\in\{1,-1\}$.
For an invertible assignment, 
the octagon after the assignment is obtained from the normalized octagon $o$ before the assignment by
\begin{equation}
\begin{array}{lll}
\sem{{\bf x}_i\,{:=}\,\epsilon{\bf x}_i+c}^\sharp\,o &=& o[(\epsilon{\bf x}_i-\epsilon c)/{\bf x}_i]	
\end{array}
\label{def:invertible}
\end{equation}
i.e., ${\bf x}_i$ in the conjunction $o$ is substituted either with ${\bf x}_i-c$ or $-{\bf x}_i+c$, respectively.
This substitution preserves normalization.
For an input assignment ${\bf x}_i\;{:=}\;?$ all information about the variable ${\bf x}_i$ is lost.
This means that
\begin{equation}
\sem{{\bf x}_i\,{:=}\,?}^\sharp\,o = o|_{-\{{\bf x}_i\}}
\label{def:input}
\end{equation}
where the octagon $o|_{-\{{\bf x}_i\}}$ is obtained from the normalized octagon $o$
by setting the upper bounds for all octagon linear forms containing ${\bf x}_i$ to $\infty$.

For a noninvertible assignment, the octagon after the assignment is obtained 
from the normalized octagon $o$ before the assignment, in two steps. 
First, all inequalities referring to the program variable ${\bf x}_i$ are removed from $o$,
i.e., their upper bounds are set to $\infty$.
Then inequalities corresponding to inequalities corresponding to the 
variable of the left-hand side are introduced. 
\begin{equation}
\begin{array}{lll}
\sem{{\bf x}_i\,{:=}\,c}^\sharp\,o	&=&	\textsf{normalize}\,(({\bf x}_i\leq c)\wedge(-{\bf x}_i\leq -c)\wedge o|_{-\{{\bf x}_i\}})	\\
\sem{{\bf x}_i\,{:=}\,\epsilon{\bf x}_j+c}^\sharp\,o &=& 
			({\bf x}_i-\epsilon{\bf x}_j\leq c)\wedge(\epsilon {\bf x}_j-{\bf x}_i\leq -c)\;\wedge\\
	&&				(o|_{-\{{\bf x}_i\}}\vee
				         o|_{-\{{\bf x}_i\}}[(\epsilon{\bf x}_i-\epsilon c)/{\bf x}_j])	
\end{array}
\label{def:non-invertible}
\end{equation}
In case that ${\bf x}_i$ is assigned a constant $c$, normalization is required to infer 
further two-variable inequalities with ${\bf x}_i$ and all other variables with known upper bounds or lower bounds.
In essence, the restrictions to guards and assignments allow defining 
the abstract semantics without resorting to \emph{linear programming}.

%
\paragraph{The system of equations.}
Subsequently, we put up a system of equations over $\D$ for octagon analysis of a program
given by its control-flow graph.
As unknowns of the system, we use the set of all pairs
$\angl{v,\ell}$, $v$ a program point and $\ell$ an octagon linear form, which is meant to receive the
upper bound for $\ell$ provided by the normalized octagon inferred for program point $v$.
For a given program point $v$, The right-hand sides for the unknowns $\angl{v,\ell}$ are determined 
as the least upper bounds over the contributions of control-flow edges $(u,\emph{act},v)$. 
These contributions are determined according to 
\eqref{def:guard},
\eqref{def:input},
\eqref{def:invertible} or
\eqref{def:non-invertible} (depending on whether \emph{act} is a guard, an input or an invertible or
non-invertible assignment) where the bounds vector of the argument octagon $o$
is retrieved from the values of the unknowns $\angl{u,\ell}$.
Having resolved the joins via a \emph{max}-policy, we are left with a system of equations
using (at least conceptually) minimum, sums and division by 2 only. 
Taking a closer look at the properties of normalization, however, we notice that division by 2 need not be
applied repeatedly. We obtain:

\begin{theorem}\label{t:octagon}
For every equation system $\Sigma$ for octagon analysis either over the integers or over the rationals,
the following holds:
\begin{enumerate}
\item	$\Sigma$ has only concave \emph{max}-iterates;
\item	Value iteration for each encountered \emph{max}-iterate $\Sigma^{\pi_t}$ terminates;
\item	\textbf{Max+Val} determines the least solution.
\end{enumerate}
\end{theorem}

\begin{proof}[Sketch]
Consider a \emph{max}-iterate $\Sigma^{\pi_{t}}$.
Instead of performing value iteration in one go, 
we first concentrate on the values for the two-variable octagon linear forms according to 
the calculation of the $a$ vector of normalization. 
As only sum and minimum are involved, both for integers or rationals, 
value iteration starting with $\infty$ will terminate with a unique solution. 
From the $b$-bounds of normalization together with these
$a$-bounds, the precise upper bounds of the one-variable octagon linear forms are obtained 
by addition and minimum. Therefore, another
finite number of rounds of value iteration suffice to determine their values.
From these, yet another finite number suffices to
determine the final results for the two-variable octagon linear forms at all program points.
Altogether, value iteration for each \emph{max}-iterate terminates and returns a unique solution.
Then, according to Lemmas \ref{l:term} and \ref{l:unique}, the third claim follows.
\qed
\end{proof}
\noindent
Value iteration in dedicated phases as in the proof of Theorem \ref{t:octagon} is not required. 
Instead, it suffices to perform any fixpoint computation on the control-flow graph starting
with $\infty$ for every unknown $\angl{v,\ell}$, and repeatedly applying the abstract octagon semantics 
until stabilization.

\subsection{Template linear constraints}
\label{ss:linear-templates}

As a further generalization of octagons, \emph{template linear constraint domains} have
been introduced \cite{DBLP:conf/vmcai/SankaranarayananCSM06}. 
%
%
Here, these are considered for bounded exact rational arithmetic only.
%
%
%
A template linear constraint domain 
is defined by a matrix $T \in {\mathbb R}^{m \times n}$.
Each abstract value from this domain is a vector $c\in\RR^m$ representing the set 
$\{x\in\RR^n \mid Tx \leq c, x\leq M, -x\leq M \}$ for some vector
$M\in\RR^n$ of sufficiently large upper bounds.
Thus, each \emph{row} $T_i$ of $T$ corresponds to one linear combination $T_i\, x$ which is bounded by the 
$i$th entry $c_i$ of $c$.
\ignore{
One such domain is the \emph{octagon} domain \cite{DBLP:journals/lisp/Mine06} where the rows $T_i$
consist of all vectors in $\{-1,0,1\}^n$ 
with at most two non-zero entries.
}
Consider a vector assignment $\textbf x\;{:=}\,A\textbf x+b$ in a program, and assume that the set of program states
reaching the assignment is described by the abstract value $c\in\RR^m$.
%
%
%
Then the vector $c' = [c'_1,\ldots,c'_m]^\top\in\RR^m$ optimally describing the set of program states after the
assignment are defined by
\begin{equation}
  {\small\begin{array}{lll} 
  c'_i &=&  |T_i|\,M\;\wedge\;
  (T_i\, b + \bigvee \{ (T_i\, A)\, x\;\mid 
          T\,x \leq c, x\leq M, -x\leq M \})
  \end{array}}
  %
  \label{eq:template-assign}
\end{equation}
with 
$|T_i| = [|t_{i1}|,\ldots,|t_{in}|]$ for
$T_i = [t_{i1},\ldots,t_{in}]$. 
Thus, each entry $c'_i$ is introduces one linear optimization problem.
Similarly for a guard $(A\textbf x\leq b)?$, we obtain the optimal bound $c'$ for the template $T$ 
after executing the guard as
\begin{equation}
  {\small c'_i = \bigvee \{ T_i\, x \mid T\,x \leq c, A\,x\leq b, x\leq M, -x\leq M \}}
  \label{eq:template-guard}
\end{equation}
All these optimization problems 
can be combined into one big optimization problem whose unknowns consist of all pairs
$\angl{u,i}$, $u$ a program point, $i\in\{1,\ldots,m\}$ ($m$ the number of rows in $T$) \cite{GawlitzaS2011}.
If $N$ is the number of program points, this optimization problem thus uses $N\cdot m$ variables.
What we obtain in this way, is a version of \textbf{Max+LP} for the template linear constraints domain.

\ignore{
%
%
Likewise, we can extend the algorithm \textbf{Max+Val} to the affine template constraints.
During value iteration, we would
Run a floating-point LP solver to solve the LP problems \eqref{eq:template-assign} and \eqref{template-guard}
whenever they are encountered. Instead of solving one huge optimization problem per encountered \emph{max}-policy, 
this algorithm solves many LP problems, each, however, with $m$ variables only.
}
To extend
the algorithm \textbf{Max+Min} to the linear template constraints domain,
we consider the respective \emph{dual} optimization problems for
\eqref{eq:template-assign} and \eqref{eq:template-guard}.
For a control-flow edge with the assignment \eqref{eq:template-assign},
from some program point $u$ to some program point $u'$ with no further incoming edge, 
the equations for the unknowns $\angl{u',i}$ take the form
\begin{equation}
{\small\begin{array}{lll}
  \angl{u',i} 	&=& |T_i| M\wedge (T_i\, b\; +\;
  		\bigwedge \{\lambda^\top (u\times\{1,\ldots,m\}) 
			+ (\lambda'+\lambda'')^\top\,M\; \mid	\\  
		& &	\lambda,\lambda',\lambda''\geq 0, 
			\lambda^\top\, T+(\lambda'-\lambda'')^\top = T_i\, A\})
		%
		%
\end{array}}
  \label{eq:template-assign-dual}
\end{equation}
where $u\times\{1,\ldots,m\}$ is the column vector $[\angl{u,1},\ldots,\angl{u,m}]^\top$.
Note that the polyhedron $P_i$ defined by 
$P_i = \{	
	\lambda,\lambda',\lambda''\geq 0\mid
	\lambda^\top\, T+(\lambda'-\lambda'')^\top = T_i\, A\}$
is \emph{independent} of the unknowns from the system.
Let $V_i = \{[v,v'_1,v''_1],\ldots,[v_r,v'_r,v''_r]\}$ denote an enumeration of the vertices of $P_i$. 
Each $[v_j,v'_j,v''_j]\in V_i$
represents one possible choice of a \emph{min}-policy for the left-hand side unknown $\angl{u',i}$.
For that choice, the equation \eqref{eq:template-assign-dual}
turns into 
\[
{\small\begin{array}{lll}
\angl{u',i} &=& T_i\, b + v_j^\top\, (u\times\{1,\ldots,m\}) + (v'_j+v''_j)^\top\,M 
\end{array}}
\]
Given a finite solution $\rho$ of this system of equations, 
the optimization problem given by the right-hand sides of \eqref{eq:template-assign-dual}
can be solved to check whether $\rho$ already provides the optimal solution.
If this is not case, a 
vertex in $V_i$
can be found which for which a strictly smaller value can be computed. That vertex 
then represents the improved \emph{min}-policy for $\angl{u',i}$. 
%
This algorithm requires to solve many optimization problems each, however, with $m+2n$ variables only.
\ignore{
Additionally, during one iteration, the current \emph{min}-policy may be improved independently in many places --- 
thus causing \emph{min}-policy iteration potentially reaching an optimum very fast.
}
%
%
%

%
%
\paragraph{Witnesses.}
%
%
If optimization problems must be solved when evaluating right-hand sides,
it is no longer trivial to verify that a given valuation $\rho$ is indeed a solution.
For each optimization sub-problem, however, in its dual form, optimality can be certified by additionally
providing a vertex where the optimal value is attained.
Likewise, when trying to improve the current \emph{max}-policy, the right-hand sides must be evaluated
to check whether a more profitable policy is at hand. Thus, given some valuation $\rho_t$ attained during
\emph{max}-policy iteration, extra optimality certificates are required to justify 
the claimed construction of $\rho'_{t+1}$, and likewise optimality certificates to justify the proposed 
next valuation $\rho_{t+1}$.

\section{Related Work}\label{sec:related}
\emph{Policy iteration} or \emph{strategy iteration} is a well-known approach for solving games \cite{Puterman_1994}.
Yet results and algorithms from game theory, in general, cannot be directly used in the context of program analysis because game theory works over transition (sub-)probability  matrices, while we work with quite different kinds
of mappings.
Plain value iteration, for instance, is commonly used for analyzing probabilistic systems where
iterations are stopped when the result does no longer seem to change much \cite{kwiatkowska:hal-00740112}.
\ignore{
\begin{quote}
The sequence of vectors $p^{\max,k}$ is guaranteed to converge eventually to $p^{\max}$. In practice, though, the computation is terminated when a pre-speciﬁed convergence criterion is met. One common approach is to check that the maximum (absolute) difference between the corresponding elements of successive vectors is below some ﬁxed threshold [\dots] Another is to check the maximum relative difference.
\end{quote}}
In contrast, value iteration for the static analysis of numerical programs is known to be impractical in many
important cases.
If (upward) value iteration, for instance, is applied to bounds analysis of the loop in Section \ref{sec:background},
the potentially long sequence of candidate upper bounds $0, 1,\ldots,N$ for variable $i$ at the loop head is attained.
In case that $N=\infty$, it even may fail to terminate.
\ignore{
that just increases an unbounded integer variable initialized to $0$, then one gets the intervals 
$[0,0]$,$[0,1]$,$[0,2]$\dots and the analysis never terminates.
}
Stopping the iteration early, however, would produce a bound that is not an inductive invariant.
\ignore{
Even when convergence within a finite number of iterations would be ensured by the boundedness of variables, it is in general impractical --- one does not want to iterated $2^{31}$ times, let alone $2^{63}$ times, to conclude that an integer variable can take any nonnegative value.
}
This motivated the introduction of widening operators to enforce termination with at least some inductive invariant 
\cite{Cousot77,DBLP:journals/logcom/CousotC92}.

The idea of applying policy iteration to abstract interpretation of numerical programs
goes back to the seminal work of Costan et al. \cite{DBLP:conf/cav/CostanGGMP05}
There, however, iteration is on \emph{min}-policies only. 
This is conceptually simpler, since each choice of an argument at a minimum results in
an \emph{over-approximation} of the original system;
iterations may be stopped at any time, providing \emph{some} solution of the system; even if iterations eventually 
terminate, there is in general no indication that the returned solution is least. 
Techniques for solving the systems without minima were not presented.
Instead, these ideas were generalized to linear
\cite{DBLP:conf/esop/GaubertGTZ07} or non-linear \cite{AdjeGG2012,adje:hal-00940804,Roux_PhD,DBLP:journals/fmsd/RouxVS18}
template-based domains. 
In these generalizations, linear or convex programming were 
employed to determine least solutions of the simplified systems.
Only in some cases, depending on the nonlinear spectral radius of some operator, 
results can be guaranteed to be minimal~\cite{adje:hal-00940804}.

\emph{Max}-policy iteration, on the other hand, was pioneered by Gawlitza \& Seidl \cite{GawlitzaS2007}.
They observed that this iteration together with plain value iteration 
for the simplified systems provides an efficient method for 
determining precise systems of polynomial integer equations  \cite{GawlitzaS2013}.
\emph{Max}-policy iteration then was applied to template linear constraints domains 
over the rationals \cite{DBLP:conf/csl/GawlitzaS07,GawlitzaS2011,GawlitzaS2013} as well as to convex relaxations of 
nonlinear template domains \cite{Gawlitza_Seidl_FMSD2013}.
Later, \emph{max}-policy iteration was combined with SMT solvers and succinct representations of sets 
to deal with programs where policies are only implicitly provided as paths through a program
\cite{DBLP:journals/corr/abs-1209-0643,DBLP:conf/sas/MonniauxS14}.
%
In all three extensions,
mathematical optimizers were employed to solve 
the encountered \emph{max}-iterates.
%
These optimizers were assumed to return mathematically precise results.
Implementation with floating-point numbers were not considered, except when convex 
relaxations were used to deal with nonlinear constraints 
\cite{Gawlitza_Seidl_FMSD2013,AdjeGG2012}. 
Roux \cite{Roux_PhD} gives a detailed account of convex relaxations in the context 
of \emph{min}- and \emph{max}-policies.
\emph{Max}-policy iteration, however, combined with \emph{min}-policy iteration for inferring numerical
program invariants has not yet been considered. Likewise, combining \emph{max}-policy iteration with
(downward) value iteration for floating-point invariants is new. 
The same holds true for our application of \textbf{Max+Val} to obtain a precise octagon analysis
over integers as well as over rationals.

\ignore{
\emph{Parametric} \emph{max}-policy iteration for systems of integer affine systems of equations
were considered in \cite{Seidl_Gawlitza_SCSS14}. There, the concepts of assignments to unknowns 
as well as of policies were generalized to decision diagrams with affine inequalities between
parameters as internal nodes. Lazy parametric policy iteration has not yet been considered, 
neither certificates on minimality for systems of affine rational equations.
}

\section{Conclusion}\label{sec:conclusion}

\emph{Max-min} policy iteration  is a promising novel technique for computing least solutions of systems of equations 
that express numerical inductive invariants in domains such as the intervals or, more generally, 
linear template constraints domains. 
Its running-time is affected by the cost of solving sparse systems of affine equations and 
the number of encountered \emph{max}- and \emph{min}-policies. 
In theory, these numbers could be exponential in the number of \emph{maximum} and \emph{minimum} operators 
occurring in the system, although we never observed that in our examples. 

Over the integers and over floating-point numbers,
\emph{max}-policy iteration can also be supplemented by value iteration to get invariants.
For integers, this results in a novel method which increases expressiveness beyond the known approach from 
\cite{GawlitzaS2007} by removing any restriction onto the set of used operators.
Over floating-point numbers, it vastly outperforms 
\emph{max}-policy iteration combined either with linear programming or \emph{min}-policy iteration.

We also clarified how to provide certificates of correctness 
and also of optimality.
Certificates should be verifiable by simple, local computations, which we made 
formal through the complexity class \NCone. We showed that, under the usual conjectures of 
complexity theory, it is insufficient to give only the computed invariants as a certificate 
of optimality. Instead, our certificate records central elements from the execution trace,
namely, the sequence of encountered \emph{max}-policies together with the corresponding valuations.

As future work, we would like to improve the efficiency of the nested \emph{min}-policy iteration 
by incrementally solving only marginally modified systems of affine equations, and
apply of our methods to enhance the precision of standard abstract interpretation based analyzers.
We envision a proof of fast convergence of the value iterations in \textbf{Max+Val}, based on spectral considerations. 
%
We also envision to exploit 
our certificates for solving \emph{parametric} systems of equations,
i.e., systems depending on external parameters.
Such systems allow representing \emph{transformations}
from parameter settings to analysis results \cite{Seidl_Gawlitza_SCSS14,DBLP:conf/sas/MonniauxS14}.

\paragraph*{Acknowledgement.} We are very indepted to the anonymous reviewers for their helpful comments 
and valuable suggestions of improvements.

\paragraph*{Data-Availability Statement.}
The artifact, including source code for our implementations, benchmarks and the evaluation scripts,
is archived and available at \href{https://dx.doi.org/10.5281/zenodo.18134107}{\textsf{10.5281/zenodo.18134107}}. 

\ignore{
\begin{credits}

\subsubsection{\discintname}
The authors have no competing interests to declare that are
relevant to the content of this article. 
\end{credits}
}
%
%

\bibliographystyle{splncs04}
\bibliography{min_max_without_lp}

\end{document}